\documentclass[sigconf]{acmart}

\usepackage[utf8]{inputenc} 
\usepackage[T1]{fontenc}    
\usepackage{hyperref}       
\usepackage{url}            
\usepackage{booktabs}       
\usepackage{amsfonts}       
\usepackage{nicefrac}       
\usepackage{microtype}      
\usepackage{xcolor}

\usepackage{algorithm, algorithmic}
\usepackage{graphicx, subfigure}
\usepackage{multirow}
\usepackage{amsmath}
\newcounter{funct}
\makeatletter
\newenvironment{funct}[1][htb]{%
  \let\c@algorithm\c@funct
  \renewcommand{\ALG@name}{Functionality}
  \begin{algorithm}[#1]%
  }{\end{algorithm}
}
\makeatother
\AtBeginDocument{%
  \providecommand\BibTeX{{%
    \normalfont B\kern-0.5em{\scshape i\kern-0.25em b}\kern-0.8em\TeX}}}

\setcopyright{acmcopyright}
\copyrightyear{2018}
\acmYear{2018}
\acmDOI{10.1145/1122445.1122456}

\begin{document}

\title{MPC-enabled Privacy-Preserving Neural Network Training against Malicious Attack}

\author{Ziyao Liu}
\affiliation{%
  \institution{Nanyang Technological University}
}
\email{ziyao002@e.ntu.edu.sg}

\author{Ivan Tjuawinata}
\affiliation{%
  \institution{Nanyang Technological University}
}
\email{ivan.tjuawinata@ntu.edu.sg}

\author{Chaoping Xing}
\affiliation{%
  \institution{Shanghai Jiaotong University}
}
\email{xingcp@sjtu.edu.cn}

\author{Kwok Yan Lam}
\affiliation{%
  \institution{Nanyang Technological University}
}
\email{kwokyan.lam@ntu.edu.sg}


\renewcommand{\shortauthors}{Liu, et al.}


\begin{abstract}
The application of secure multiparty computation (MPC) in machine learning, especially privacy-preserving neural network training, has attracted tremendous attention from the research community in recent years. MPC enables several data owners to jointly train a neural network while preserving the data privacy of each participant. However, most of the previous works focus on semi-honest threat model that cannot withstand fraudulent messages sent by malicious participants. In this paper, we propose an approach for constructing efficient $n$-party protocols for secure neural network training that can provide security for all honest participants even when a majority of the parties are malicious. Compared to the other designs that provide semi-honest security in a dishonest majority setting, our actively secure neural network training incurs affordable efficiency overheads of around 2X and 2.7X in LAN and WAN settings, respectively. Besides, we propose a scheme to allow additive shares defined over an integer ring  $\mathbb{Z}_N$ to be securely converted to additive shares over a finite field $\mathbb{Z}_Q$, which may be of independent interest. Such conversion scheme is essential in securely and correctly converting shared Beaver triples defined over an integer ring generated in the preprocessing phase to triples defined over a field to be used in the calculation in the online phase.
\end{abstract}




\maketitle

\section{Introduction}
\label{Intro}
During this last decade, with the development of machine learning, especially deep neural network (DNN), the scenario where different parties, e.g., data owners or cloud service providers, jointly solve a machine learning problem while preserving their data privacy has attracted tremendous attention from both academia and industry. Federated learning scheme seems to offer the possibility for distributed privacy-preserving machine learning, focusing on the cross-device and cross-silo setting where multiple clients train local models using their raw data, then aggregate their models under the coordination of a central server \cite{kairouz2019advances}. However, there is still no formal privacy guarantee for this baseline learning model. Therefore, secure multiparty computation (MPC) as a practical, mature privacy-preserving technique aiming to enable multiple parties to jointly evaluating a function is natural to be applied to address such privacy issues of training neural networks in a distributed manner.

\textbf{Previous Results: }
There have been several schemes proposed to perform a distributed neural network. This research direction is arguably pioneered by the design of SecureML by Mohassel and Zhang\cite{mohassel2017secureml}, where they propose several MPC-friendly activation functions to enable neural network training on secret shared data. The study of the neural network on shared data can be classified into several classes based on their goals, the number of parties, and the threat model. First, some models focus on neural network prediction. A scheme designed for neural network prediction called MiniONN \cite{liu2017oblivious} was constructed based on the design of SecureML\cite{mohassel2017secureml}. With the extensive use of packing techniques and additive homomorphic encryption (AHE) cryptosystem along with garbled circuit, Gazelle\cite{juvekar2018gazelle} provides a neural network prediction protocol with a more efficient linear computation. 

Secondly, some other models provide secure neural network training for two active parties providing security against one corrupted party. As discussed above, SecureML \cite{mohassel2017secureml} provides us with a neural network training protocol for two parties, which is secure against a semi-honest adversary controlling one party. 
Among the existing designs, ABY~\cite{demmler2015aby} presents a framework of efficient conversion between various two-party computation schemes to support different machine learning algorithms. It is proven to be secure against a semi-honest adversary controlling one party. SecureNN~\cite{wagh2019securenn} relies on more sophisticated protocols and up to two non-colluding external parties to provide neural network training protocols for two active parties. Its security is guaranteed against a semi-honest adversary controlling up to 1 party. 
Note that most of the designs in this situation are designed for two parties with only a semi-honest security guarantee, with the number of corrupted parties being at most half of the total number of participating parties. Despite some efforts to extend these designs to a larger number of active parties, the improvement has instead been limited.

Lastly, there have also been some designs dedicated to the system with more than two parties with active security guarantee in dishonest majority setting. Most of these works are based on the SPDZ scheme~\cite{damgaard2012multiparty}. In such works, SPDZ is used to provide several accurate and efficient machine learning algorithms. However, since SPDZ is designed to be a general-purpose MPC scheme, most of the libraries implementing SPDZ does not offer primitives and optimizations for the purpose of neural network training.
Therefore, in this work, we present dedicated MPC protocols based on SPDZ for convolutional neural networks (CNN) training and demonstrate that our protocols obtain active security with affordable overheads compared to the existing secure neural network training in a semi-honest setting.

For SPDZ protocol, the main efficiency improvement and active security come from the extensive use of pre-computed Beaver triples with Message authentication code (MAC) (see Section \ref{SPDZ}) to accelerate arithmetic operations. During this last decade, following the initial scheme proposed in BDOZ \cite{bendlin2011semi}, many researchers have been working on protocols of efficient Beaver triple generation in a malicious setting, which are based on either homomorphic encryption schemes such as the original SPDZ~\cite{damgaard2012multiparty} or oblivious transfer such as MASCOT~\cite{keller2016mascot} or SPDZ2k~\cite{cramerspdz2k}. 
Relying on these schemes, all parties can jointly generate Beaver triples over a finite field $\mathbb{F}_Q$ or a ring $\mathbb{Z}_{2^k}$, which can be directly used for MAC checking in the online phase of SPDZ and its variants. Specifically, for HE-based schemes, we have to choose a proper crypto-system to be used in the offline phase.
A popular cryptosystem to be used in such a situation is the leveled BGV scheme, which has high performance due to the extensive use of packing techniques, e.g., single-instruction multiple data (SIMD) trick.

Unfortunately, Beaver triples generated based on AHE crypto-systems over $\mathbb{Z}_N$, e.g., Paillier \cite{paillier1999public} 
as used in SecureML \cite{mohassel2017secureml}, cannot be directly used for the verification of standard SPDZ, which is based on a finite field. Therefore, an instance of SPDZ that is based on Paillier or DGK requires a secure scheme to transform the triples generated modulo $\mathbb{Z}_N$ to the underlying field of the SPDZ, $\mathbb{Z}_Q.$

\textbf{Our contributions: } In this work, we propose the construction of efficient $n$-party protocols for secure CNN training in malicious majority setting, including linear and convolutional layer, Rectified Linear Unit (ReLU) layer, Maxpool layer, normalization layer, dropout layer, and their derivatives. In addition, we present a secure conversion scheme for shares defined over an integer ring $\mathbb{Z}_N$ to shares over a prime field $\mathbb{Z}_Q$, which can also be used to convert shared Beaver triples correctly. We believe that this result may be of independent interest. Our experimental results show that our protocols for secure neural network training provide affordable overheads of around 2X and 2.5X in LAN and WAN settings, respectively, compared with existing schemes in the semi-honest setting.

\textbf{Organisation of the paper: } The rest of the paper is organized as follows. In Section \ref{Preliminaries}, we provide notations and threat model used in this paper, as well as a brief discussion of secure computation and neural network. In Section \ref{SupProtocol}, we introduce several supporting protocols, including Distributed Paillier crypto-system, SPDZ protocol, and protocols of secure computation of fixed-point numbers. Section \ref{ProposedProtocols} contains our MPC protocols, which can be used to construct an efficient, secure neural network protocol, as illustrated in Section \ref{NNMPC}. Then we analyze the performance of our protocol in Section \ref{ComR}. Finally, we present our experimental evaluation in Section \ref{Exp} and give a conclusion in Section \ref{Conclu}.

\section{Preliminaries}
\label{Preliminaries}

\subsection{Neural Network Training}
\label{NN}

A neural network consists of many layers with nodes defined by a set of linear operations, such as addition and multiplication, and non-linear operations such as ReLU, Maxpool, and dropout. At a very high level, we represent a neural network as a function $\mathbf{w} = f(\mathbf{x})$ where $\mathbf{x}$ represents the set of input data with their respective labels, $\mathbf{w}$ represents the set of weights of the neural network and function $f$ can be represented with linear operations and non-linear operations as mentioned above. The target of training a neural network is to obtain these weights $\mathbf{w}$, which can be used to map a new unlabeled data $\mathbf{x}^\ast$ to its predicted label $\mathbf{\ell}^\ast$, i.e., prediction.

\subsection{Secure Multiparty Computation}
\label{SecureComput}

Privacy-preserving technology provides a privacy guarantee for data for various purposes, such as for computation or for publishing. This technology broadly encompasses all schemes for privacy-preserving function evaluations, including but not limited to differential privacy (DP), secure multiparty computation (MPC), and homomorphic encryption (HE). It is well known that DP provides a tradeoff between accuracy and privacy that can be mathematically analyzed while MPC and HE offer cryptographic privacy but with high communication or computation overheads. In this work, we focus on MPC to construct efficient neural network training protocols. Thanks to both theoretical and engineering breakthroughs, MPC has moved from pure theoretical interests to practical implementations.

In addition to pure homomorphic encryption-based MPC, there are two schemes that can be used to construct MPC protocols, i.e.,  circuit garbling and secret sharing. Circuit garbling, as used in SecureML \cite{mohassel2017secureml}, involves encrypting and decrypting keys in a specific order, while the latter emulates a function evaluation more efficiently based on the "secretly-shared" inputs between all parties. Our work leverages an additive secret sharing MPC protocol called SPDZ (see Section \ref{SPDZ}), such that for each data $x$, it will be randomly split up into $n$ pieces and then be distributed among $n$ parties (see Algorithm \ref{alg:Resharing} in Section \ref{SPDZ}). In the rest of the paper, we write $[x]^t = [\left \langle x_1 \right \rangle ^ t, ..., \left \langle x_n \right \rangle ^ t]$ to denote  $x\in\mathbb{Z}_t$ being secretly shared between all parties such that each $P_i$ holds $\langle x_i\rangle^t$ and $x=\sum_{i=1}^n\langle x_i\rangle^t.$ For simplicity of notation, when the context is clear, we abuse the notation and use $x_i$ instead of $\langle x_i\rangle^t$, the share of $x$ owned by party $P_i$. Similarly, when the underlying space $\mathbb{Z}_t$ is clear from the context, we write $[x]$ instead of $[x]^t.$

\subsection{Threat Model and Security}
\label{ThreatModel}

In many real-life neural network applications, the training data are distributed across multiple parties, which are independent business entities and are required to comply with the applicable data privacy regulations. Therefore, due to the competitive nature of business organizations, in this work, we consider the scenario where a majority of parties may collude to obtain the data from the other parties by sending fraudulent messages to them. Such a threat model is the same as the one used in SPDZ, i.e., security against a malicious adversary controlling up to $n-1$ parties. This means that in an $n$-party setting, such MPC protocol is secure even if $n-1$ parties are corrupted by a malicious adversary. Such threat model is different from that of MPC protocols in SecureML \cite{mohassel2017secureml} and SecureNN \cite{wagh2019securenn}, which are against semi-honest adversary. As demonstrated in \cite{cramer2015secure}, semi-honest protocols can be elevated into the malicious model, which may incur infeasible cost overhead. However, thanks to the online-offline architecture of SPDZ, such overhead can be moved from the online phase to the offline phase, and thus the amortized efficiency of function evaluation can be improved. Our security definition is based on the Universal Composability (UC) framework, and we refer interested readers to \cite{canetti2001universally} for the details. 

Intuitively, to show the security of our protocol $\Pi,$, we first define an ideal functionality $\mathcal{F}$, which performs the same calculation of $\Pi$ but with the existence of a trusted third party to help in the computation truthfully. $\Pi$ is said to securely realize $\mathcal{F}$ in the presence of an adversary $\mathcal{A}$ if for any possible adversary $\mathcal{A}$ against the protocol $\Pi,$ there exists a simulator $\mathcal{S}$ which interacts with the ideal functionality $\mathcal{F}$ such that given the interactions of $\mathcal{A}$ with $\Pi$ and $\mathcal{S}$ with $\mathcal{F},$ there is no environment in which the two interactions can be distinguished. Furthermore, we are also using a hybrid model in our security proof. Intuitively such a model is especially useful when a protocol $\Pi$ is built based on other protocol $\Pi^\prime.$ Suppose that $\Pi$ and $\Pi^\prime$ are supposed to simulate $\mathcal{F}$ and $\mathcal{F}^\prime$ respectively. We say $\Pi$ securely realizes $\mathcal{F}$ in $\mathcal{F}^\prime$ hybrid model if the protocol $\Pi$ where any instances of $\Pi^\prime$ is replaced by $\mathcal{F}^\prime$ securely realizes $\mathcal{F}.$ We note that this implies that $\Pi$ itself is secure if combined with the proof that $\Pi^\prime$ securely realizes $\mathcal{F}^\prime.$

The correctness and security of our proposed protocol depend on the supporting building blocks, i.e., distributed Paillier crypto-system, SPDZ, and secure computation of fixed-point numbers. Data representation follows the format in \cite{catrina2010secure}, and we use the standard SPDZ scheme over a finite field $\mathbb{Z}_Q$. The subprotocol of SPDZ involved in this works include protocol for data resharing $\mathrm{Resharing(\cdot)}$, multiplication $\mathrm{MulTri(\cdot)}$, Paillier based Beaver triple generation $\mathrm{TriGen(\cdot)}$, and MAC checking. The details of all the above mentioned supporting protocols can be found in Section \ref{SupProtocol}.

\section{Supporting Protocols}
\label{SupProtocol}
\subsection{Distributed Paillier Cryptosystem}
\label{DistributedPaillier}

Paillier \cite{paillier1999public} is a public key encryption scheme that possesses partial homomorphic property. The public key is $N = p \cdot q$ and the secret key is $(p,q)$ pair where $p$ and $q$ are large primes. First, we fix $g$ to be a random invertible integer modulo $N^2$. The encryption of a message $m$ is defined to be $c=E(m)=g^{m} r^{N}\pmod {N^{2}}$ for a randomly chosen invertible $r \in \mathbb{Z}_N$. The decryption function is $D(c)=\frac{L\left(c^{\lambda}\pmod{N^{2}}\right)}{L\left(g^{\lambda}\pmod{ N^{2}}\right)}\pmod{N}$ where function $L$ is defined as $L(x)=\frac{x-1}{N}$ 
and $\lambda = \varphi(N)$ where $\varphi$ is Euler's totient function. Paillier supports homomorphic addition between two ciphertexts and homomorphic multiplication between a plaintext and a ciphertext, in particular, $E(m_1) \cdot E(m_2)\pmod {N^2} = E(m_1 + m_2 \pmod {N})$ and $E(m_1)^{m_2}\pmod {N^2} = E(m_1 \cdot m_2 \pmod {N})$. For simplicity, we denote the following two functions: $E(m_1+m_2\pmod N) = PAdd(E(m_1),E(m_2))$ and $E(m_1\cdot m_2 \pmod N)= PMult(E(m_1),m_2)$. We can easily generalize the two notations such that $E(m_1+m_2+\cdots+m_r \pmod N)=PAdd(E(m_1),E(m_2),\cdots, E(m_r))$ and $E(m_1\cdot m_2 \cdot \cdots \cdot m_r\pmod N)=PMult(E(m_1),m_2,\cdots, m_r)$. Due to the invertibility of $g$ and $r$ modulo $N^2,$ it is easy to see that a Paillier ciphertext is invertible modulo $N^2.$ Hence, if $c=E(m),$ we will also have $c^{-1}\pmod{N^2}$ to be a valid encryption of $-m.$ We denote $PInv(E(m))= E(-m\pmod{N})\equiv E(m)^{-1}\pmod{N^2}.$ Using $PInv,$ Paillier can support homomorphic subtraction, $E(m_1-m_2\pmod{N})=PAdd(E(m_1),PInv(E(m_2)))$. These homomorphic properties enable several protocols proposed in Section \ref{SPDZ} and Section \ref{ProposedProtocols}. However, in terms of MPC scenario, the secret key pair is not allowed to be owned by any party. To keep the hardness of composite residuosity used for the security of Paillier, the value of $p$ and $q$ can not be known by anyone. Hence we need to generate the public key in distributed manner and keep its factors secret while still enabling joint decryption to be done without revealing the private key. Fortunately, such distributed Paillier does exist. Distributed Paillier key generation includes two sub-protocols, i.e., (i) distributed RSA modulus generation, and (ii) distributed biprimality test to verify the validity of generated RSA modulus in (i). We aim to have a secure distributed Paillier cryptosystem which provides security against malicious adversary controlling up to $n-1$ parties in an $n$-party setting. Thus we rely on the scheme proposed in~\cite{hazay2019efficient} to achieve this.
We denote the Paillier cryptosystem with plaintext space $\mathbb{Z}_{N}$ as $Paillier_{N}$, as well as its encryption and distributed decryption, i.e., $Enc_{N}(\cdot)$, $Dec_{N}(\cdot)$.

\subsection{SPDZ}
\label{SPDZ}

SPDZ is a well-known secret sharing based MPC protocol against malicious majority proposed in \cite{damgaard2012multiparty}. Following this initial somewhat homomorphic encryption (SHE) based work, several variants are proposed, see, for example, ~\cite{keller2016mascot,cramerspdz2k}.
We will refer to all of these variants in the SPDZ family as SPDZ.

SPDZ consists of a pre-processing or offline phase, which is independent of both the input data and the very efficient online phase for function evaluation. In the offline phase, all parties jointly generate some "raw materials", typically the Beaver triples. In the online phase, the parties only need to exchange some shares and perform some efficient verification. The active security is guaranteed by the MAC, which enables the validation of parties' behavior during computation. In the rest of this section, several important techniques in SHE based SPDZ are introduced in order to construct some higher-level protocols proposed in Section \ref{ProposedProtocols}.

\textbf{Data Resharing: } Given $Enc_{t}(x)$, all parties can follow the protocol given in Algorithm \ref{alg:Resharing} to obtain $[x]^t$ as the shares of $x$. We note that this resharing of an encrypted value is only done during pre-processing phase to help in the generation of auxiliary values, and it is not used in the sharing protocol of private inputs of the function $f.$ For the sharing of private inputs during the online phase of the computation, it follows the protocol given in Algorithm~\ref{alg:Sharing}.

\begin{algorithm}[htbp]
\caption{Data Resharing: $[x]^t \leftarrow \mathrm{Resharing}(Enc_{t}(x))$}
\label{alg:Resharing}
\begin{algorithmic}[1]
\STATE Each party $P_{j}$ publishes $Enc_{t}(r_j)$, where $r_j$ is uniformly selected from $\mathbb{Z}_t;$
\STATE All parties calculate $Enc_t(r+x)=PAdd(Enc_t(r_1),Enc_t(r_2),\cdots, Enc_t(r_n),Enc_t(x))$ using homomorphic addition;
\STATE All parties jointly decrypt $Enc_{t}(r + x)$ to obtain $r+x;$
\STATE $P_1$ sets its share $x_1 = r+x - r_1$, $P_j$ sets its share $x_j = - r_j$ for $j \neq 1;$
\STATE Return $[x]^t;$
\end{algorithmic}
\end{algorithm}

\begin{algorithm}[htbp]
\caption{Data Sharing: $[x]^t \leftarrow \mathrm{Share}(x)$ where $x$ is a private value owned by $P_i$}
\label{alg:Sharing}
\begin{algorithmic}[1]
\REQUIRE A shared random value $[r]^t$
\STATE Each party $P_j$ sends his share of $[r]^t, \langle r\rangle^t_j$ to $P_i$ enabling $P_i$ to recover the value of $r;$
\STATE $P_i$ sets $\langle x\rangle^t_i= x+r-\langle r\rangle^t_i$ and for $j\neq i, P_j$ sets $\langle x\rangle^t_i=-\langle r\rangle^t_j;$
\STATE Return $[x]^t=(\langle x\rangle^t_1,\cdots, \langle x\rangle^t_n);$
\end{algorithmic}
\end{algorithm}

\textbf{Arithmetic operation: } SPDZ is based on secret sharing. Thus there is no communication cost for addition and scaling by a public constant. Multiplication between two secretly shared values is more complex, but we can use the well known Beaver triple trick to accelerate this operation. 
For completeness, we provide the protocol for multiplying two secretly shared values using Beaver triple.
In the following discussion, all values are secretly shared using the additive secret sharing scheme over the same space. Assume that we have generated three secret shared values $[a]$, $[b]$, and $[c]$, called Beaver triples, such that $c = a \cdot b$. Given $[x]$ and $[y]$, all parties can follow the protocol in Algorithm \ref{alg:MulTri} to calculate $[x \cdot y]$. 
Note that all the protocols in SPDZ can be applied to matrices.

\begin{algorithm}[htbp]
\caption{Multiplication based on Beaver triple: $[x \cdot y] \leftarrow \mathrm{MulTri}([x], [y], [a], [b], [c])$}
\label{alg:MulTri}
\begin{algorithmic}[1]
\STATE Each party $P_{i}$ publishes $x_i - a_i$ and $y_i - b_i;$
\STATE Each party $P_{i}$ compute $x - a$ and $y - b;$
\STATE $P_1$ sets its share $z_1 = c_1 + (x-a) \cdot b_1 + (y-b) \cdot a_1 + (x-a)\cdot(y-b)$, $P_j$ sets its share $z_j = c_j + (x-a) \cdot b_j + (y-b) \cdot a_j$ for $j \neq 1;$
\STATE Return $[z];$
\end{algorithmic}
\end{algorithm}

Considering that one triple cannot be used to perform two multiplications for privacy reason, 
the number of triples we need to generate depends on the number of multiplications we want to complete. Furthermore, these Beaver triples do not depend on inputs data as well as the function to be evaluated, which means they can be generated at any point prior to evaluating the function, i.e., the offline phase in SPDZ, thus enabling a highly efficient online phase.

\textbf{Beaver triple generation: } Algorithm \ref{alg:TriGen} describes the $n$-party protocol for Beaver triple generation based on Paillier such that $[a]^N, [b]^N, [c]^N \leftarrow \mathrm{TriGen}()$. For simplicity of discussion, we only discuss the protocols under the semi-honest setting. As discussed before, such protocols can be made secure against a malicious adversary by the combination of zero-knowledge proof and the standard technique of sacrificing an auxiliary value to check the correctness of another. A more detailed discussion of this technique can be found in\cite{damgaard2012multiparty}. As mentioned in Section \ref{Intro}, these Beaver triples generated using Algorithm \ref{alg:TriGen} are over $\mathbb{Z}_N$, thus cannot be directly in the online phase of SPDZ, which is over a finite field.

\begin{algorithm}[htbp]
\caption{Beaver triple generation based on Paillier: $[a]^N, [b]^N, [c]^N \leftarrow \mathrm{TriGen}()$}
\label{alg:TriGen}
\begin{algorithmic}[1]
\STATE Each party $P_{i}$ publishes $Enc_{N}(a_i)$ and $Enc_{N}(b_i)$, where $a_i$ and $b_i$ are uniformly selected from $\mathbb{Z}_N;$
\STATE Each party $P_i$ computes $Enc_N(a)=PAdd(Enc_N(a_1),\cdots, Enc_N(a_n));$
\STATE Each party $P_i$ computes and publishes $Enc_{N}(a \cdot b_i) = PMult(Enc_N(a),b_i)= Enc_{N}(a)^{b_i}$ using homomorphic multiplication in Paillier;
\STATE Each party $P_i$ computes $Enc_{N}(c) = PAdd(Enc_N(a\cdot b_1),\cdots,Enc_N(a,b_n));$
\STATE All parties call $\mathrm{Resharing}(Enc_{N}(c))$ to get $[c]^N;$
\STATE Return $[a]^N, [b]^N, [c]^N;$
\end{algorithmic}
\end{algorithm}

Note that in step 3 of Algorithm \ref{alg:TriGen}, only $P_i$ knows the value of $b_i$, which enables the homomorphic multiplication between a plaintext and ciphertext (refer Section \ref{DistributedPaillier}), i.e., $Enc_{N}(a)$ and $b_i$, thus no information leaks.

\textbf{MAC checking: } To obtain active security over $\mathbb{Z}_Q$, the main idea of SPDZ is to use unconditional MAC, which enables verification of computation correctness. This authentication scheme prevents parties from cheating on their interactive computation with high probability. In SPDZ, to enable authentication, each private value, including the Beaver triples that are generated, comes with their respective tags. To obtain this, first, the parties agree on a random MAC key $\alpha\in \mathbb{Z}_Q$ which is secretly shared among all parties. To compute the tag of a secretly shared value $[x]^Q$, the parties compute $[\alpha\cdot x]^Q$ and store it along with their shares of $[x]^Q$. We can observe that if some adversaries cheat such that the secretly shared value is changed from $x$ to $x^{\prime}$, they can do so undetected only if they can modify the corresponding tag $[m]$ to $[m^{\prime}]$ such that $m^{\prime} = \alpha.x^{\prime}$. This means that the probability of cheating without being detected is equal to the probability of guessing $\alpha$ correctly, which is inversely proportional to the finite field size $Q.$
When considering a similar scheme over a ring $\mathbb{Z}_N,$ the security is no longer as strong. This is due to the fact that, contrary to $\mathbb{Z}_Q,$ not all non-zero value in $\mathbb{Z}_N$ is invertible. Because of this, the probability that $m^{\prime} = \alpha.x^{\prime}$ becomes larger. For example, if $N=2^k$ and $\alpha=2^{k-1}$, $m^{\prime}$ can only be either 0 or 1 making the probability 1/2. As illustrated above, we use Paillier to generate Beaver triples with MAC, which means all the secret shared values are in $\mathbb{Z}_N$. Hence we have to convert all these shares from $\mathbb{Z}_N$ to $\mathbb{Z}_Q$ while preserving the relationship between them. Note that for any sub-protocols, all the inputs and outputs should always be secretly shared and not be in the clear. In addition to the shares of the outputs $[y]^Q$, the parties should also hold a secret share of their respective tags $[\alpha\cdot y]^Q$.

\subsection{Secure Computation of Fixed Point Numbers}
\label{DataRep}

For typical neural networks, data and weights are represented by floating-point numbers. However, in terms of combining neural networks with cryptographic techniques such as HE and MPC, we have to use a very large finite field to preserve full accuracy.
This method supports only a limited number of multiplications to avoid the overflow, which is prohibitive for the neural network where a large number of multiplications are involved. 
In order to avoid a large increase in complexity, we may consider limiting the precision by the use of fixed-point arithmetic and truncation protocol. To achieve this, we extend the protocols provided in~\cite{catrina2010secure} to $n$-party. Such a method is guaranteed to be correct by~\cite{catrina2010secure} while the active security is achieved by the use of the MAC-checking method in SPDZ.

\textbf{Data representation: } Real numbers can be treated as a sequence of digits including integer and fractional parts split by a radix point. More specifically, for any real number $\tilde{x}$, set $e$ and $f$ as positive integers such that $|\tilde{x}|<2^{e-1}$ and the storage accuracy is within $2^{-f}.$ Then we can find a sign bit $s\in\{0,1\}$ and $d_{-f},\cdots, d_{e-2}\in\{0,1\}$ such that $\tilde{x} = (-1)^s \cdot\sum_{i=-f}^{e-2} d_i 2^i.$ To encode $\tilde{x},$ we first encode it as an integer $\bar{x}$ by multiplying it by $2^f.$ Hence we have $\bar{x} = (-1)^s \sum_{i=0}^{e+f-2} d_{i-f} 2^i = 2^f \tilde{x}.$ Next, we set $Q$ to be a prime number that is at least $k+\kappa$ bits where $k=e+f$ and $\kappa$ is the security parameter. To encode $\tilde{x}$ as a field element in $\mathbb{Z}_Q,$ we map $\tilde{x}$ to the element $\bar{x}\in \mathbb{Z}_Q.$ Calculation can then be done using MPC schemes that is over $\mathbb{Z}_Q$. 

\textbf{Truncation and comparison: } In order to maintain the same resolution of secretly shared values and enable comparison, two truncation protocols are used in our work, i.e., probabilistic truncation $\mathrm{TruncPr}$ and deterministic truncation $\mathrm{Trunc}$, as given in \cite{catrina2010secure}. Probabilistic truncation supports efficient truncation as no bit-wise operation is involved, and some "raw materials" needed in the protocol can be prepared during the pre-processing phase. However, it introduces error with a probability depending on the size of the Least Significant Bit (LSB) after truncation. In terms of the data representation of $x$, the Most Significant Bit (MSB) of $x$ determines whether $x$ is greater than $0$ or not, which can be obtained by simply truncating the last $k-1$ least significant bits—since the LSB after truncation is now large, using TruncPr for this purpose will yield a non-negligible error. Hence, an alternative truncation protocol is required. Deterministic truncation is less efficient but enables truncation with zero error probability. Therefore, although probabilistic truncation may be used to avoid overflow during multiplication computation, deterministic truncation is needed for comparison computation. We denote by $\mathrm{GEZ}$, an adapted version of comparison protocol $\mathrm{LTZ}$ following the notation given in \cite{catrina2010improved} such that $\mathrm{GEZ}([x], k) = 1$ if $x \geq 0$, and 0 otherwise, in order to keep the consistency with ReLU function (see Section \ref{NN}).

\textbf{Arithmetic operations: } Addition and public scaling on additive shares can be done without interaction while maintaining the same resolution. Multiplication can be done using Beaver triple method (see Section \ref{SPDZ}) with its resolution changing from $2^{-f}$ to  $2^{-2f}$, which means probabilistic truncation is needed. Protocols of the division with public divisor and secretly shared divisor are also given in \cite{catrina2010secure} that offers reasonable accuracy and efficiency.

\section{Proposed Protocols}
\label{ProposedProtocols}

In this section, we describe protocols to support Beaver triple conversion and neural network training. We assume that there are two distributed Paillier crypto-systems with different plaintext space $\mathbb{Z}_N$ and $\mathbb{Z}_{N^{\prime}}$. Based on the definition of Paillier and the notation introduced in Section \ref{DistributedPaillier}, let $N = p \cdot q$ where $2^{\kappa} < p, q < 2^{\kappa + 1}$, and $N^{\prime} = p^{\prime} \cdot q^{\prime}$ where $2^{2\kappa} < p^{\prime}, q^{\prime} < 2^{2\kappa+1}$, such that $N^{\prime} > N^2$. All the multiplications involved in these protocols can be done following a similar multiplication protocol in Paillier based Beaver triple generation given in Section \ref{SPDZ}. In the rest of the paper, we write $[x]^t = [\left \langle x_1 \right \rangle ^ t, \cdots, \left \langle x_n \right \rangle ^ t]$ to denote  $x\in\mathbb{Z}_t$ being secretly shared between all parties such that each $P_i$ holds $\langle x_i\rangle^t$ and $x\equiv\sum_{i=1}^n\langle x_i\rangle^t\pmod{t}.$ For simplicity of notation, when the context is clear, we abuse the notation and use $x_i$ instead of $\langle x_i\rangle^t$, the share of $x$ owned by party $P_i$. Similarly, when the underlying space $\mathbb{Z}_t$ is clear, we write $[x]$ instead of $[x]^t.$

\subsection{Comparison Modulo $N$} \label{subsec:comp}

The first supporting protocol that we want to introduce is the secure comparison protocol $\mathrm{GEZ}$ over $\mathbb{Z}_N.$ More specifically, this function receives a secretly shared value $[x]^N$ and the bit length $k$ of $x.$ It then outputs $[s]^N$ where $s=1$ if $x\geq 0$ and $0$ otherwise. Note that this function can be used to compare two secretly shared values $[x]^N,[y]^N$ by computing $\mathrm{GEZ}([x-y]^N,k).$ This algorithm is an adapted version of $\mathrm{LTZC}$ protocol in \cite{securescm}, which is  based on the following remark: for x with $k$-bit length (refer to data representation in Section \ref{DataRep}), if $x < 0$ then $\lfloor x/2^{k-1} \rfloor = 1$, and if $x \geq 0$ then $\lfloor x/2^{k-1} \rfloor = 0$. Note that the protocol $\mathrm{GEZ}$ is essentially the same as the $\mathrm{LTZC}$ protocol, but its supporting sub-protocols including random bit generation $\mathrm{RndBit}$, random integer generation $\mathrm{PRndInt}$ and inequality test $\mathrm{BitLTC}$, are required to be defined over ring rather than finite field. Thus, $\mathrm{LTZC}$ cannot be directly applied to define $\mathrm{GEZ}$ due to the difference in the underlying space. In order to have a similar protocol to $\mathrm{LTZC}$ over $\mathbb{Z}_N,$ we first describe some sub-protocols that are required, namely random bit generation, random invertible integer generation, and random integer generation modulo $\mathbb{Z}_N.$

Algorithm \ref{alg:RndBit} describes our $n$-party protocol for random bit generation over $\mathbb{Z}_N.$ The standard protocol in $\mathbb{Z}_Q$ setting, such as $\mathrm{RAN_2()}$ in \cite{damgaard2006unconditionally}, 
relies on the fact that the number of roots of the quadratic polynomial $x^2-a$ over the field $\mathbb{Z}_Q$ is exactly $2$ for any non-zero $a\in \mathbb{Z}_Q.$
However, this may no longer work when working modulo $N.$ More specifically, for any given non-zero quadratic residue $a\in \mathbb{Z}_N,$ there exactly $4$ elements $x_1,x_2,x_3,x_4\in\mathbb{Z}_N$ such that $x_i^2\equiv a\pmod{N}$ for $i=1,\cdots, 4.$ Instead, we will be using XOR to produce a shared random bit from the random bits generated by the $n$ parties. To simplify the discussion of $\mathrm{RndBit},$ we assume that $n=2^\ell$ for some positive integer $\ell.$ The full specification of $\mathrm{RndBit}$ in this special case can be found in Algorithm~\ref{alg:RndBit}. The full specification of the functionality as well as claims for its correctness and security are discussed in Appendices~\ref{subapp:compfunct} and~\ref{subapp:compcorsec} respectively.

\begin{algorithm}[htbp]
\caption{Random bit generation over $\mathbb{Z}_N$: $[a]^N \leftarrow \mathrm{RndBit}(N)$}
\label{alg:RndBit}
\begin{algorithmic}[1]
\STATE Each party $P_i$ generates a uniformly random bit $a^{(0)}_i\in \{0,1\}$ and secretly shares it to other parties $\left[a^{(0)}_i\right]^N;$
\FOR{$d=1,\cdots,\ell$}
\FOR{$i=0,\cdots, 2^{\ell-d}-1$}
\STATE The parties calculate $\left[a^{(d)}_i\right]^N=\left[a^{(d-1)}_{2i}\right]^N+\left[a^{(d-1)}_{2i+1}\right]^N-2\cdot\left[a^{(d-1)}_{2i}\right]^N\cdot \left[a^{(d-1)}_{2i+1}\right]^N;$
\ENDFOR
\ENDFOR
\STATE The parties calculate $[check]^N=\left[a_0^{(\ell)}\right]^N\cdot \left(1-\left[a_0^{(\ell)}\right]^N\right)$ and reveal $check;$
\IF{$check\neq 0$}
\STATE The protocol is aborted without any output;
\ELSE
\STATE Output $[a]^N=\left[a_0^{(\ell)}\right]^N;$
\ENDIF
\end{algorithmic}
\end{algorithm}

Algorithm \ref{alg:RndInv} describes our $n$-party protocol for random integer with inverse generation over $\mathbb{Z}_N$, which is an adapted version of $\mathrm{PRandInv}$ in \cite{securescm}. Algorithm~\ref{alg:RndInv} is built on another secure protocol $\mathrm{RndInt}(N),$ which is a secure protocol realizing the functionality $\mathcal{F}_{RndInt}.$ We note that $\mathcal{F}_{RndInt}$ is a functionality that takes an RSA modulus $N$ as input and outputs a uniformly random element of $\mathbb{Z}_N.$ This can be done by letting each party $P_i$ deal a sharing $[r_i]^N$ and $[r]^N$ is defined to be $[r]^N=\sum_{i=1}^N[r_i]^N.$ Note that as long as there is one honest party, the resulting $[r]^N$ can be proved to be uniformly distributed in $\mathbb{Z}^N$. The correctness of the algorithm is straightforward since $u$ is invertible if and only if $x$ and $y$ are also invertible. Security proof is similar to that of $\mathrm{PRandInv}$ in~\cite{securescm}.

\begin{algorithm}[htbp]
\caption{Random integer with inverse generation over $\mathbb{Z}^{N}$: $([r]^{N}, [r^{-1}]^{N})\leftarrow \mathrm{RndInv}(N)$}
\label{alg:RndInv}
\begin{algorithmic}[1]
\STATE All parties call $[x]^N \leftarrow \mathcal{F}_{RndInt}(\lceil\log N\rceil, N), [y]^N \leftarrow \mathcal{F}_{RndInt}(\lceil\log N\rceil, N)$, calculate $[u]^N\leftarrow [x]^N\cdot [y]^N$ and then output $u;$
\STATE Repeat step $1$ until $u$ is invertible;
\STATE Return $([x]^N,u^{-1}[y]^N);$
\end{algorithmic}
\end{algorithm}




Apart from the supporting sub-protocols we have discussed above, in order to define the protocol $\mathrm{GEZ}$ which is the main objective of this subsection, there are other sub-protocols that need to be defined. Such sub-protocols have also been defined in~\cite{securescm} and they can be adopted to our situation with little to no modification. Firstly, we define the functionality $\mathcal{F}_{PreMulC}.$ This functionality receives $\ell$ secretly shared values $[a_1]^N,\cdots, [a_\ell]^N$ as inputs and outputs $\ell$ secretly shared values $([b_1]^N,\cdots, [b_\ell]^N)$ where $b_i=\prod_{j=1}^i a_j.$ Note that such functionality can be realized by the help of $\mathcal{F}_{RndInv}.$ Since the protocol $\mathrm{PreMulC}$ can be defined exactly in the same way as the protocol with the same name in~\cite{securescm}, we omit the definition of $\mathrm{PreMulC}$ and its security definition in the $\mathcal{F}_{RndInv}$-hybrid model.

The second functionality we need is $\mathcal{F}_{Inner}.$ This functionality accepts a pair of vectors of secretly shared values of the same length $([a_1]^N,\cdots, [a_\ell]^N),([b_1]^N,\cdots, [b_\ell]^N)$ and outputs a secretly shared value $[c]^N$ where $c=\sum_{i=1}^\ell a_ib_i.$ Note that this functionality can again be realized by the use of $\ell$ multiplication operations that is done in parallel. Due to the simplicity of the specification of the protocol $\mathrm{Inner}$ realizing $\mathcal{F}_{Inner},$ we again omit its definition as well as its security definition. 

The last functionality we need is $\mathcal{F}_{PRndInt}.$ This functionality accepts two integers $k$ and $N$ as inputs where $N$ is an RSA modulus and $k<\log N.$ Having such inputs, the functionality returns an integer of size $k$ bits and secretly shared over $\mathbb{Z}^N.$ Note that such functionality can be realized by the use of $k$ parallel calls of $\mathcal{F}_{RndBit}$ to obtain the $k$ bits to be used as the binary representation of the output. Hence, it is easy to see that such protocol $\mathrm{PRndInt}$ securely realizes $\mathcal{F}_{PRndInt}$ in the $\mathcal{F}_{RndBit}$-hybrid model.

Now we are ready to construct discuss the protocol $\mathrm{BitLTC}$ which is the main supporting protocol required to define our main objective of this section, $\mathrm{GEZ}.$ The protocol $\mathrm{BitLTC}$ performs a comparison between a public integer $a$ with its binary representation $a=\sum_{i=0}^{k-1} a_i 2^i$ and a bit-wise secretly shared integer $b=\sum_{i=0}^{k-1} b_i 2^i$ defined over $\mathbb{Z}_N$ where $N$ is an RSA modulus as given in Algorithm \ref{alg:BitLTC}. The protocol $\mathrm{BitLTC}$ takes $(a,([b_{k-1}]^N,\cdots, [b_0]^N)$ as input and returns $[s]^N$ where $s=1$ if $a<b$ and $s=0$ otherwise. The protocol $\mathrm{BitLTC}$ is defined in a very similar manner to a protocol with the same name defined in~\cite{securescm}. Hence we will only briefly discuss the intuition behind the protocol while the functionality and security guarantee in the $(\mathcal{F}_{PreMulC}, \mathcal{F}_{Inner},\mathcal{F}_{RndBit},\mathcal{F}_{PRndInt})$-hybrid model directly follow from~\cite{securescm}.


The protocol $\mathrm{BitLTC}$ is mainly divided to two steps. First, the scheme computes a secretly shared integer $[y]^N$ such that its least significant bit is $1$ if $a<b$ and it is $0$ otherwise. This is done by first calculating $y_i=b_i(1-a_i)$ which is $1$ if and only if $a_i<b_i.$ Let $i^\ast$ be the largest index such that $y_i=1.$ Hence for any $i>i^\ast,$ we have $y_i=0.$ Note that $a<b$ if and only if $a_{i^\ast}<b_{i^\ast}$ while the values of $a_j$ and $b_j$ do not matter for any $j<i.$ Such observation is utilized in the algorithm in the definition of $x_i.$ It is easy to see that for any $j<i^\ast,$ we have $x_j$ to be a positive power of $2.$ This shows that $x_j y_j$ is even for any $j\neq i^\ast$ while $x_{i^\ast} y_{i^\ast} = (a_i<b_i).$ Hence it is easy to see that by the calculation of $y,$ its least significant bit is indeed $1$ if and only if $a<b.$ The first half of $\mathrm{BitLTC}$ is summarised in the algorithm $\mathrm{BitLTMap}$ which can be found in Algorithm~\ref{alg:BitLTMap}.

The second half of the algorithm takes $[y]^N$ as input and extracts its least significant bit $[s]^N.$ This is done by first masking $y$ with a random mask of size $k+\kappa$ bits $2^{k-1} + 2r + u$ with $r$ being a $k+\kappa-1$ bits integer and $u\in\{0,1\}$ to get a random value $c.$ By definition the least significant bit of $c$ is $c_0=s\oplus u$ where $s$ is the least significant bit of $y$ that we want to extract. Hence, $s$ can be obtained by calculating $s=c_0\oplus u,$ as has been performed in Algorithm~\ref{alg:LSB}.

\begin{algorithm}[htbp]
\caption{Inequality test between public integer and bit-wise secretly shared integer over $\mathbb{Z}^{N}$: $[s]^{N}\leftarrow \mathrm{BitLTC}(a,([b_{k-1}]^N,\cdots, [b_0]^N)$}
\label{alg:BitLTC}
\begin{algorithmic}[1]
\STATE $[y]^N \leftarrow \mathrm{BitLTMap}(a,([b_{k-1}]^N,\cdots, [b_0]^N);$
\STATE $[s]^N \leftarrow \mathrm{LSB}([y]^N,k);$
\STATE Return $[s]^N;$
\end{algorithmic}
\end{algorithm}

\begin{algorithm}[htbp]
\caption{Bit Mapping of inequality test $[y]^N \leftarrow \mathrm{BitLTMap}(a,([b_{k-1}]^N,\cdots, [b_0]^N)$}
\label{alg:BitLTMap}
\begin{algorithmic}[1]
\FOR{$i=1,\cdots,k-1$}
\STATE $[d_i]^N \leftarrow a_i+[b_i]^N-2a_i[b_i]^N;$
\ENDFOR
\STATE $\left(\left[x_{k-2}\right]^N, \ldots,\left[x_{0}\right]^N\right) \leftarrow \mathcal{F}_{PreMulC}\left(\left[d_{k-1}\right]^N+1, \ldots,\right.$ $\left.\left[d_{1}\right]^N+1\right);$
\FOR{$i=0,\cdots,k-1$}
\STATE $[y_i]^N \leftarrow [b_i]^N(1-a_i);$
\ENDFOR
\STATE $[y]^N \leftarrow\left[y_{k-1}\right]^N+\mathcal{F}_{Inner}\left(\left(\left[x_{0}\right]^N,\ldots,\left[x_{k-2}\right]^N\right),\right.$ $\left.\left(\left[y_{0}\right]^N, \ldots,\left[y_{k-2}\right]^N\right)\right);$
\STATE Return $[y]^N;$
\end{algorithmic}
\end{algorithm}

\begin{algorithm}[htbp]
\caption{Least significant bit extraction $[v]^N \leftarrow \mathrm{LSB}([y]^N,k)$}
\label{alg:LSB}
\begin{algorithmic}[1]
\STATE $[u]^N \leftarrow \mathcal{F}_{RndBit}(N);$
\STATE $[r]^N \leftarrow \mathcal{F}_{PRndInt}(k+\kappa -1);$
\STATE $c \leftarrow \mathrm{Output} (2^{k-1}+[y]^N+2[r]^N+[u]^N);$
\STATE $[v]^N \leftarrow c_0+[u]^N-2c_0[u]^N;$
\STATE Return $[v]^N;$
\end{algorithmic}
\end{algorithm}

Now we are ready to present our $n$-party protocol $\mathrm{GEZ}$ for comparison over $\mathbb{Z}_N$, which is an adapted version of $\mathrm{LTZC}$ protocol in \cite{securescm}. To keep the consistency with the ReLU function used in neural network, we flip the output from $\mathrm{LTZC}.$ Intuitively, $\mathrm{GEZ}$ is done by first masking the $k$-bit private input $x$ with a $k+\kappa$-bit random secret mask $r$ where the binary representation of $r\pmod{2^{k-1}}$ is known. Once $c=x+r$ is published and $c'\equiv c\pmod{2^{k-1}}$ is known, $c'$ and $r'=r\pmod{2^{k-1}}$ are then used as inputs for $\mathcal{F}_{BitLTC}$ to check if $c'<r'.$ This is used to obtain $x'=x\pmod{2^{k-1}}$ which can then be used to extract the most significant bit of $x,$ which contains the information of the sign of $x,$ as required. The full specification of $\mathrm{GEZ}$ can be found in Algorithm~\ref{alg:GEZ}. It is easy to see that the security of $\mathrm{GEZ}$ directly follows from that of $\mathcal{F}_{RndBit},\mathcal{F}_{PRndInt},$ and $\mathcal{F}_{BitLTC}.$ The summary of the functionality and the security claims can be found in Appendices~\ref{subapp:compfunct} and~\ref{subapp:compcorsec} respectively.


\begin{algorithm}[htbp]
\caption{Comparison over $\mathbb{Z}_N$: $[s]^N \leftarrow \mathrm{GEZ}([x]^N,k)$}
\label{alg:GEZ}
\begin{algorithmic}[1]
\STATE For each $i \in [0, ..., k-1]$, all parties call $\left[r_{i}\right]^N = \mathcal{F}_{RndBit}(N)$ in parallel, and thus obtain $\left[r^{\prime}\right]^N = \sum_{i=0}^{k-2} 2^{i} \cdot\left[r_{i}\right]^N;$
\STATE All parties call $\left[r^{\prime \prime}\right]^N = \mathcal{F}_{PRndInt}(\kappa + 1,N);$
\STATE All parties publish $c \leftarrow\mathrm{Output}(2^{k-1} \cdot\left[r^{\prime \prime}\right]^N+\left[r^{\prime}\right]^N + 2^{k-1}+[x]^N)$, and then calculate $c^{\prime} = c \pmod{2^{k-1}};$
\STATE All parties call $[u]^N = \mathcal{F}_{BitLTC}\left(c^{\prime}, \left(\left[\left.\left.r_{k-2}^{\prime}\right]^N, \cdots, \right[r_{0}^{\prime}\right]^N\right)\right);$
\STATE All parties compute $[x^{\prime}]^N = c^{\prime} - \left[r^{\prime}\right]^N + 2^{k-1}[u]^N;$
\STATE All parties compute $[s]^N = 1 + \left([x]^N-\left[x^{\prime}\right]^N\right) (2^{-k+1}\pmod{N});$
\end{algorithmic}
\end{algorithm}

\subsection{Wrap, Modulo Reduction, Share Conversion}
\label{MSW}
In this section, we discuss the secure conversion protocol that will help us in converting the values we generated during the offline phase (modulo $N$ for some RSA modulus $N>n$) to the equivalent value that is compatible with the online phase (modulo $Q$ for a prime $Q$). More specifically, given $[a]^N,$ the additive share of a secret value $a$ modulo $N,$ we want to calculate $[a]^{Q},$ the additive share of the same secret value modulo $Q$ for some prime $Q.$ First, for simplicity, we discuss the transformation of the secret-sharing values. Note that initially, we want our secret value and its shares to be an element in $S_1=\left\{-\frac{n-1}{2},\cdots, \frac{n-1}{2}\right\}.$ For simplicity of our argument in this section, we transform all these values to be non-negative value in $S_2=\{0,\cdots, n-1\}$ via congruence operation. Note that this does not change the correctness of any sharing, and transformation between the two formats can be done trivially.

Suppose that $[x]^N=( x_1^N,\cdots,  x_n^N).$ Then there exists an integer $\delta\in\{0,\cdots, n-2,n-1\}$ such that 
\begin{equation}\label{deltasource}
x= x_1^N+\cdots+x_n^N-\delta N.    
\end{equation}
Hence if we want to consider the equation modulo $Q,$ we will have 
\begin{equation}\label{deltamod}
x \equiv 
\begin{array}{l}
(x_1^N \pmod{Q})+\cdots+(x_n^N \pmod{Q})\\
- (\delta \pmod{Q})\cdot (N\pmod{Q}) 
\end{array}
\pmod{Q}.
\end{equation}

So in order to calculate $[x]^{Q}$ from $[x]^N,$ we need to calculate the value of $\delta$ which can be rewritten as $\delta = \left\lfloor \frac{\sum_{i=1}^n x_i^N}{N}\right\rfloor$.

Now we discuss how we can calculate the value of $\delta.$ Note that Equation~\eqref{deltasource} will not yield the value of $\delta$ if it is computed modulo $N.$ Intuitively, if we consider the equation modulo $N'$ for some $N'$ such that $N'>N^2,$ the relation between the two sides are now equality instead of equivalence modulo $N'$ and hence we can use it to calculate $\delta.$ Once we have the equation modulo $N',$ we can find the maximum value of $j$ such that $\sum_{i=1}^n x_{i}^N-jN\geq 0.$ It is easy to see that $\delta=\sum_{j=0}^{n-1} (\sum_{i=1}^n x_i^N-jN\geq 0 ).$ Now since the equation is modulo $N',$ the calculation will give us $[\delta]^{N'}.$ We let this procedure to be called $[\delta]^{N'}\leftarrow\mathrm{LiftWrap}([x]^N,N')$ which is only applicable if $N'>N^2>n^2.$ Algorithm~\ref{alg:LiftWrap} provides the complete $\mathrm{LiftWrap}$ protocol. The correctness and the security of the protocol $\mathrm{LiftWrap}$ can be found in Appendix~\ref{subapp:MSWcorsec}.

\begin{algorithm}[htbp]
\caption{LiftWrap: $[\delta]^{N'} \leftarrow \mathrm{LiftWrap}_{N}([x]^N, N^{\prime})$}
\label{alg:LiftWrap}
\begin{algorithmic}[1]
\STATE Each party $P_{i}$ computes $\left \langle {x^{\prime}} \right \rangle _i^{N^{\prime}} = \left \langle {x} \right \rangle _i^N\pmod{N^{\prime}};$
\STATE For each $i \in \{1, ..., n-1\}$, all parties call the functionality $\left[\delta_{i}\right]^{N'}= \mathcal{F}_{GEZ} \left( \left[x^{\prime}\right]^{N^{\prime}}-\left[i \cdot N\right]^{N^{\prime}}, l\right)$, where $l= \lceil \log _{2} N^{\prime} \rceil;$
\STATE Return $\left[\delta\right]^{N'} = \sum_{i=1}^{n-1} \left[\delta_{i}\right]^{N'};$
\end{algorithmic}
\end{algorithm}

The next step is to convert $[\delta]^{N'}$ to $[\delta]^N.$ In other words, we need a secure conversion protocol $\mathrm{DropMod}$ to convert a secretly shared value $[x]^{N'}$ back to $[x]^{N}$ where $N'>N^2>n^2.$ In order to complete this, first, we observe that given $[\delta]^{N'}=(\delta_1,\cdots, \delta_n),$ setting $y_i\equiv x_i-\delta_i N \pmod{N'},$ we have $\sum_{i=1}^n y_i = x \pmod{N'}.$ In other words, for any $[x]^N,$ we can calculate $[x]^{N'}.$ Let this procedure be called $[x]^{N'}\leftarrow \mathrm{LiftMod}([x]^N,N')$ which is only applicable if $N'>N^2>n^2.$ Algorithm~\ref{alg:LiftMod} provides the complete $\mathrm{LiftMod}$ protocol. The security is guaranteed based on the security guarantee of $\mathrm{LiftWrap}$ protocol. The summary of the functionality and the security claims can be found in Appendices~\ref{subapp:MSWfunct} and~\ref{subapp:MSWcorsec} respectively. 

\begin{algorithm}[htbp]
\caption{Lift shares in $\mathbb{Z}_N$ to $\mathbb{Z}_{N'}$: $[x]^{N'} \leftarrow \mathrm{LiftMod}_{N}([x]^N, N^{\prime})$}
\label{alg:LiftMod}
\begin{algorithmic}[1]
\STATE Parties jointly call the functionality $[\delta]^{N'}=(\delta'_1,\cdots, \delta'_n)= \mathcal{F}_{LiftWrap}([x]^N,N');$
\STATE For each $i\in\{1,\cdots,n\},$ having $x_i$ and $\delta'_i$ (the shares of $[x]^N$ and $[\delta]^{N'}$ respectively), $P_i$ calculates $x'_i\equiv x_i-\delta'_i N \pmod{N'};$
\STATE Return $[x]^{N'} = (x'_1,\cdots, x'_n);$
\end{algorithmic}
\end{algorithm}

Now we are ready to discuss the last subprotocol needed for the Wrap function, $[x]^N\leftarrow\mathrm{DropMod}([x]^{N^{'}},N).$ Intuitively, The protocol $\mathrm{DropMod}$ is done by using the help of a random value $r$ that is secretly shared twice, $[r]^N$ and $[r]^{N'}$ such that $x+r<N'.$ Having such $r,$ we can compute and reveal $[y]^{N'}=[x+r]^{N'}.$ Having $y,$ we can then calculate $[x]^N=(y\pmod{N})-[r]^N.$ Note that the requirement that $x+r<N'$ is required to avoid having any wrap-arounds in the equation $y\equiv x+r\pmod{N'}$ which is essential in the correctness of the equality $[x]^N=(y\pmod{N})-[r]^N.$ However, because of the absence of any wrap-arounds, the reveal of $y$ may leak some information about $x.$ To avoid such leakage, as has been shown in~\cite{catrina2010improved}, we require $r$ to be at least $2^{\kappa}$ times larger than $x$ for a statistical security with security parameter $\kappa.$ Since we assumed that $x<N,$ this can be done by making sure that the number of bits of $r$ is $\lceil \log N\rceil +\kappa.$ This gives another requirement on the size of $N',$ namely, we require $N'>(2^\kappa N)^2.$ The full protocol $\mathrm{DropMod}$ can be found in Algorithm~\ref{alg:newDropMod}. The functionality and the correctness as well as the security claims can be found in Appendices~\ref{subapp:MSWfunct} and~\ref{subapp:MSWcorsec} respectively.


\begin{algorithm}[htbp]
\caption{Convert shares in $\mathbb{Z}_{N^{'}}$ to $\mathbb{Z}_{N}$: $[x]^{N} \leftarrow \mathrm{DropMod}_{N}([x]^{N^{'}}, N)$}
\label{alg:newDropMod}
\begin{algorithmic}[1]

\STATE Parties jointly call $\mathcal{F}_{RndBit}$ to generate a random bit $[b_i]^N=\mathcal{F}_{RndBit}(N)$ for $i=0,\cdots, \lceil \log(N)\rceil+\kappa-1;$
\STATE Parties jointly call $[b_i]^{N^{'}}=\mathcal{F}_{LiftMod}([b_i]^N,N^{'})$ for $i=0,\cdots, \lceil \log(N)\rceil+\kappa-1;$
\STATE Parties locally compute $[r]^{N^{'}}=\sum_{i=0}^{\lceil \log(N)\rceil+\kappa-1}2^i[b_i]^{N^{'}};$
\STATE Parties locally compute $[(r\pmod{N})]^N=\sum_{i=0}^{\lceil \log(N)\rceil+\kappa-1} 2^{i} \pmod{N} [b_i]^{N};$
\STATE Parties locally compute $[y]^{N^{'}}=[x]^{N^{'}}+[r]^{N^{'}},$ publish their shares of $[y]^{N^{'}}$ and recover $y;$
\STATE Return $[x]^{N} = y \pmod{N}-[r]^N;$

\end{algorithmic}
\end{algorithm}



Note that we can apply $\mathrm{DropMod}$ to obtain $[\delta]^N$ from $[\delta]^{N'}.$ We note that in this use, $\mathrm{DropMod}$ is secure since $\delta<n<N.$ So this also guarantees the security of the protocol $[\delta]^N\leftarrow \mathrm{Wrap}_N([x]^N,N').$ Algorithm~\ref{alg:Wrap} provides the complete $\mathrm{Wrap}_N$ protocol. The full specification on the functionality and claims for its correctness and security are briefly discussed in Appendices~\ref{subapp:MSWfunct} and~\ref{subapp:MSWcorsec} respectively.

\begin{algorithm}[htbp]
\caption{Wrap over $\mathbb{Z}_N$: $[\delta]^N \leftarrow \mathrm{Wrap}_{N}([x]^N)$}
\label{alg:Wrap}
\begin{algorithmic}[1]
\STATE Parties agree on a security parameter $\kappa$ and a positive integer $N'$ such that $N'>(2^\kappa N)^2$ where $N'$ is either an RSA modulus or a prime; 
\STATE Parties jointly call $[\delta]^{N'}=\mathcal{F}_{LiftWrap}([x]^N,N');$
\STATE Parties jointly call and return $[\delta]^N= \mathcal{F}_{DropMod}([\delta]^{N'},N);$
\end{algorithmic}
\end{algorithm}

Note that the conversion protocols $\mathrm{LiftMod}$ and $\mathrm{DropMod}$ are only securely and correctly applicable in a very restrictive case. More specifically, $\mathrm{LiftMod}$ can only convert from $\mathbb{Z}_S$ to $\mathbb{Z}_{S'}$ where $S'>S^2>n^2$ and $\mathrm{GEZ}$ must be well defined over $\mathbb{Z}_S.$ In other words, $S$ must be either an RSA modulus or a prime. On the other hand, $\mathrm{DropMod}$ can only convert from $\mathbb{Z}_{S'}$ to $\mathbb{Z}_S$ with $S$ and $S'$ having the same requirements as the ones in $\mathrm{LiftMod}.$ Furthermore, $\mathrm{DropMod}_S([x]^{S'},S)$ is only guaranteed to be correct when $x<S.$

Recall that our main objective of this part is to have a secure conversion protocol to convert a secret sharing $[x]^N$ to $[x]^Q$ where $N$ is an RSA modulus while $Q$ is a prime. Since this needs to be used to convert secret sharing of random values or Beaver Triple generated modulo $N,$ in order to have all possible random values modulo $Q,$ we need to have $N>Q.$ Note that if we use $\mathrm{DropMod}$ for this purpose, the value of $N$ needs to be much bigger than $Q,$ more specifically, $N>Q^2.$ In the following, we propose another conversion protocol $\mathrm{ShaConv}$, which can accomplish this goal securely as long as $N>Q.$

First we recall Equation~\eqref{deltamod}, which is essential in our discussion of $\mathrm{ShaConv}.$ 
\begin{equation*}
x \equiv 
\begin{array}{l}
(x_1^N \pmod{Q})+\cdots+(x_n^N \pmod{Q})\\
- (\delta \pmod{Q})\cdot (N\pmod{Q}) 
\end{array}
\pmod{Q}.
\end{equation*}

Having $[x]^N=(x_1,\cdots, x_n),$ the first $n$ terms of the equation above can be calculated locally by each party $P_i.$ Now in order for the conversion to be completed, we need the last term, $(\delta \pmod{Q})\cdot (N\pmod{Q}).$ Recall that the only information we have about $\delta$ is its secret sharing modulo $N, [\delta]^N$ through the use of $\mathcal{F}_{Wrap}.$ Note that to be able to calculate $[\delta]^Q\cdot (N\pmod{Q}),$ we need to first convert $[\delta]^N$ to $[\delta]^Q$ which can be achieved by using a variant of $\mathrm{DropMod}.$ However, this can only be achieved securely if $\delta<Q.$ In order to guarantee this, in our discussion, we will assume that $n<Q<N<\sqrt{N'}.$ Now suppose that we have $[\delta]^N$, and we would like to calculate $[\delta]^Q.$ Since we do not have the guarantee that $Q<N^2$, we cannot apply $\mathrm{DropMod}$ directly. Instead, we will again use the space $\mathbb{Z}^{N'}$ for this purpose. More specifically, after the calculation of $\mathrm{LiftWrap}$ to obtain $[\delta]^{N'},$ we can directly call $\mathrm{DropMod}$ to obtain $[\delta]^Q$ instead of $[\delta]^{N}.$ Now once $[\delta]^Q$ is obtained, we can obtain $[\delta]^Q\cdot (N\pmod{Q})$ completing the calculation of Equation~\eqref{deltamod}. It is easy to see that since all the sub-protocols being used here are secure, the protocol that calculates Equation~\eqref{deltamod} we have just discussed is secure. The complete protocol of $\mathrm{ShaConv}$ can be found in Algorithm~\ref{alg:ShaConv}.

\begin{algorithm}[htbp]
\caption{Share conversion: $[x]^{Q} \leftarrow \mathrm{ShaConv}([x]^N,Q,(N'))$}
\label{alg:ShaConv}
\begin{algorithmic}[1]
\IF{$N'$ is not included in the input}
\STATE Parties agree on a positive integer $N'$ such that $N'>N^2$ where $N'$ is either an RSA modulus or a prime;
\ENDIF
\STATE Parties jointly call $[\delta]^{N'}=\mathcal{F}_{LiftWrap}([x]^N,N');$
\STATE Parties jointly call $[\delta]^Q=(\delta_1,\cdots, \delta_n)=\mathcal{F}_{DropMod}([\delta]^{N'},Q);$
\FOR{$i=1,\cdots,n$}
\STATE $P_i$ possesses $x_i$ and $\delta_i,$ the $i$-th share of $[x]^N$ and $[\delta]^Q$ respectively;
\STATE $P_i$ locally computes $x'_i\equiv x_i - \delta_i\cdot (N\pmod{Q}) \pmod{Q};$
\ENDFOR
\STATE Return $[x]^Q=(x'_1,\cdots,x'_n);$
\end{algorithmic}
\end{algorithm} 
The functionality and the correctness as well as security claims can be found in Appendices~\ref{subapp:MSWfunct} and~\ref{subapp:MSWcorsec} respectively.

Note that Algorithm~\ref{alg:ShaConv} can be used for any $Q$ and $N$ securely as long as they satisfy the following requirements: (i) $n<Q<N.$, (ii) we have a secure $\mathrm{GEZ}$ protocol modulo $N$, and (iii) we have secure $\mathrm{RndInt}$ and $\mathrm{GEZ}$ protocols modulo $Q$. Due to this observation, we have that $\mathrm{ShaConv}$ is applicable as long as $Q$ and $N$ are either RSA moduli or prime numbers.

\subsection{Beaver Triple Conversion}
\label{TripleConv}


In this section, we focus on the effort of converting values secretly shared over $\mathbb{Z}_N$ for some RSA modulus $N$ to a prime field $\mathbb{Z}_Q.$ We aim to have a method to enable us also to convert Beaver triples. This is especially useful to convert the Beaver triples generated in Algorithm~\ref{alg:TriGen} to a form that can be used in our online phase, which is defined over the prime field $\mathbb{Z}_Q.$ Instead of relying on bit decomposition and bit sharing conversion over a field, our method relies on multiple instances of Paillier cryptosystems.

First we observe that given a triple $([a]^N,[b]^N,[c]^N),$ we have $c\equiv ab\pmod{N}$ or equivalently, $ab=c+\sigma N$ for some integer $\sigma$ such that $|\sigma|\leq n-1.$ Similar to the discussion of $\delta$ in the previous section, in order to get the value of $\sigma,$ we need to lift the equation modulo $N'$ for some $N'>N^2>n^2.$ Using the algorithm $\mathrm{LiftMod}$ described above, we can obtain $([a]^{N'},[b]^{N'},[c]^{N'}).$ Then $[\sigma N]^{N'}=[ab-c]^{N'}.$ Note that it is not secure to use $\mathrm{DropMod}$ to obtain $[\sigma N]^Q$ from $[\sigma N]^{N'}$ even if $N'$ is an RSA modulus or a prime. This is because it is impossible that $\sigma N<Q.$ Hence we will need to use $\mathrm{ShaConv}$ to achieve this. In order to make this possible, we require that the $N'$ we choose to be either an RSA modulus or a prime. Once we have $[\sigma]^Q,$ it is easy to see that $[ab]^Q\equiv [c]^Q + [\delta N]^Q \pmod{Q}.$ This protocol, denoted by $\mathrm{TripConv}$ is secure due to the security of all the sub-protocols involved. The complete protocol of $\mathrm{TripConv}$ can be found in Algorithm~\ref{alg:TripConv}.

\begin{algorithm}[htbp]
\caption{Beaver triple conversion: $([a^{\prime}]^Q, [b^{\prime}]^Q, [c^{\prime}]^Q) \leftarrow \mathrm{TripConv}(([a]^N, [b]^N ,[c]^N),Q)$}
\label{alg:TripConv}
\begin{algorithmic}[1]
\STATE Parties agree on $N'$ such that $N'>N^2$ and $N'$ is either an RSA modulus or a prime;
\STATE Parties jointly call the functionality $([a]^Q,[b]^Q,[c]^Q)=\mathcal{F}_{ShaConv}(([a]^N,[b]^N,[c]^N),Q);$
\STATE Parties jointly call the functionality $([a]^{N'},[b]^{N'},[c]^{N'}) = \mathcal{F}_{LiftMod}([a]^N,[b]^N,[c]^N,N');$
\STATE Parties jointly compute $[\sigma N]^{N'}=[a]^{N'}\cdot[b]^{N'} - [c]^{N'};$
\STATE Parties jointly compute $[\sigma N]^Q=\mathcal{F}_{ShaConv}([\sigma N]^{N'},Q);$
\STATE Set $[c']^Q= [c]^Q+[\delta N]^Q;$
\STATE Return $([a]^Q,[b]^Q,[c']^Q);$
\end{algorithmic}
\end{algorithm}
The functionality and the correctness as well as security claims can be found in Appendices~\ref{subapp:BTCfunct} and~\ref{subapp:BTCcorsec} respectively.

\subsection{Probabilistic Bit Generation}\label{ProbBitGen}
In this section, we discuss the protocol $\mathrm{PrRndBit},$ a protocol with input $p\in(0,1)$ and a prime $Q$ to output $[b]^Q$ where $b=0$ with probability $p$ and $b=1$ with probability $1-p.$ Algorithm \ref{alg:PrRndBit} describes our $n$-party protocol for probabilistic random bit generation over $\mathbb{Z}_Q$ such that $[b]^Q = \mathrm{PrRndBit}(Q,p).$ The generated bit share can be used for computation in dropout layer in neural network. To simplify calculation, in this calculation, instead of considering an element $x\in \mathbb{Z}_Q$ as an integer belonging to $\{-(Q-1)/2,\cdots, (Q-1)/2\},$ we consider it as a non-negative integer belonging to $\{0,\cdots, Q-1\}$ where the transformation is done by adding $Q$ to elements corresponding to negative integers in the former representation. Note that for a uniformly sampled $a \in \mathbb{Z}_Q,$ we have $a - \lfloor p \cdot Q \rfloor < 0$ with probability approximately $p$, and $a - \lfloor p \cdot Q \rfloor \geq 0$ with probability approximately $1-p$. Therefore, $\mathrm{GEZ}_{Q}([a - \lfloor p \cdot Q \rfloor], l) = 0$ with approximate probability $p$ and $\mathrm{GEZ}_{Q}([a - \lfloor p \cdot Q \rfloor], l) = 1$ with approximate probability $1-p$. Note that we can have a more accurate probability by using a larger $Q.$

\begin{algorithm}[htbp]
\caption{Probabilistic bit generation over $\mathbb{Z}_Q$: $[b]^Q \leftarrow \mathrm{PrRndBit}(Q,p)$}
\label{alg:PrRndBit}
\begin{algorithmic}[1]
\STATE All parties deal a random sharing $[a]^Q$, where $a \in \mathbb{Z}_Q;$
\STATE All parties call $[b]^Q = \mathcal{F}_{GEZ}([a]^Q-\lfloor p\cdot Q\rfloor, \ell)$, where $\ell= \lceil \log _{2} Q \rceil;$
\STATE Return $[b]^Q;$
\end{algorithmic}
\end{algorithm}

The functionality and the correctness as well as security claims can be found in Appendices~\ref{subapp:PBGfunct} and~\ref{subapp:PBGcorsec} respectively.

\section{MPC for Neural Network}
\label{NNMPC}
In this section, we describe various protocols to support efficient secure neural network training based on protocols given in Section \ref{SPDZ}, Section \ref{DataRep}, and Section \ref{ProposedProtocols}. Our protocols focus on $n$-party setting where correctness and security are guaranteed by our supporting protocols. Compared with MPC based neural network protocols in SecureML \cite{mohassel2017secureml} and SecureNN \cite{wagh2019securenn}, our protocols are applicable to a larger number of parties. Furthermore, compared to SecureNN~\cite{wagh2019securenn}, our protocols do not require external parties to assist the computation. Note that all the secret shares in this section are over a finite field $\mathbb{Z}_Q$.

\subsection{Linear and Convolutional Layer}
Since operations in linear layer and convolutional layer are exactly multiply-and-accumulates on matrix, all parties can jointly call addition and multiplication protocols in SPDZ to make an efficient evaluation. Note that multiplications on matrix rely on matrix Beaver triples such that $\mathbf{a}_{r \times s} \cdot \mathbf{b}_{s \times t} = \mathbf{c}_{r \times t}$, where $\mathbf{a}_{r \times s}, \mathbf{b}_{s \times t}, \mathbf{c}_{r \times t}$ are matrix. Indeed, matrix Beaver triple generation involves extra multiply-and-accumulates compared with that of single Beaver triple, hence takes more time. However, this can be done during the offline phase, thus greatly improves the efficiency of evaluating multiply-and-accumulates in the online phase.

\subsection{ReLU with Derivative}
In a neural network, the ReLU function is a function that depends on the non-negativity of the input such that

$$\mathrm{ReLU}(x)=
\begin{cases}
x& {x \geq 0} \\
0& {x < 0}
\end{cases}$$

Therefore, evaluating the ReLU function boils down to a comparison between $x$ and $0$.

In addition, it is easy to see that the derivative of $\mathrm{ReLU},$ denoted by $\mathrm{DReLU},$ can be formulated as follows. 

$$\mathrm{DReLU}(x)=
\begin{cases}
1& {x \geq 0} \\
0& {x < 0}
\end{cases}$$

Therefore we can conclude that for any matrix $\mathbf{x}$ of any size, $\mathrm{DReLU(\mathbf{x})}= \mathbf{s}= (\mathbf{x}\geq 0)$ where the comparison is done entry-wise and $\mathrm{ReLU}(\mathbf{x})= \mathbf{x} \times \mathbf{s}$ where $\times$ is an entry-wise matrix multiplication. Following this argument, the parties can then consecutively calculate $\mathrm{ReLU}$ and $\mathrm{DReLU}$ given a secretly shared input $[x]^Q$ following the protocol described in Algorithm~\ref{alg:ReLU}. We note that the security of this protocol is guaranteed by the security of all the sub-protocols being used during its calculation.

\begin{algorithm}[htbp]
\caption{ReLU: $([\mathbf{y}], [\mathbf{s}])\leftarrow \mathrm{ReLU}([\mathbf{x}])$}
\label{alg:ReLU}
\begin{algorithmic}[1]
\STATE All parties call $\mathrm{GEZ}([\mathbf{x}])$ to obtain $[\mathbf{s}] = [\mathbf{x} \geq 0];$
\STATE All parties call $\mathrm{MulTrip}([\mathbf{x}], [\mathbf{s}])$ to obtain $[\mathbf{y}] = [\mathbf{x}] \times [\mathbf{s}]$ where $\times$ represents entry-wise matrix multiplication;
\STATE Return $[\mathbf{y}]$ and $[\mathbf{s}];$
\end{algorithmic}
\end{algorithm}

\subsection{Maxpool with Derivative}

Maxpool is a layer of neural network that outputs the maximum values of various submatrices of the input matrix determined by several parameters, namely, filter size and stride. Since each submatrix can be handled independently in parallel, we focus on finding the maximum value of a submatrix, which can be represented as a list of $s$ elements. To find such maximum value, we use the divide and conquer strategy where the comparison can be made in $\log s$ rounds. To simplify the description of the protocol, we first assume that $s=2^p$ for some positive integer $p.$ In each round, we can pair up the elements and keep the larger element for the next round of comparison. This way, the number of elements to be compared in each round is reduced by half from the previous round. This can be done until we are left with one element, which is the largest element required as the output of $\mathrm{Maxpool}.$

In order to enable backward propagation, we will need the derivative of $\mathrm{Maxpool},$ which we denote by $\mathrm{DMaxpool}.$ Suppose that given an input list $\mathbf{x}=(x_1,\cdots, x_s)$ with $\mathrm{Maxpool}(\mathbf{x})=x^\ast$ where $x^{\ast}$ is the $i^\ast$-th entry of $\mathbf{x}.$ Then $\mathrm{DMaxpool}(\mathbf{x})=\mathbf{v}=(v_1,\cdots, v_s)$ where $v_{i^\ast}=1$ and $v_{j}=0$ for all other $j.$  it is easy to see that the intermediate comparison results from the protocol $\mathrm{Maxpool}$ can be used to provide ``path'' from the maximum value to the $x_i$ which is the maximum value. So multiplying all the intermediate comparison results in the path from the maximum value to any of the values will return $0$ if the value is not the maximum value while it will be $1$ in exactly one of the paths, as required. Our complete protocols for $\mathrm{Maxpool}$ and $\mathrm{DMaxpool}$ can be found in Algorithms~\ref{alg:Maxpool} and~\ref{alg:DMaxpool} respectively. Here $\mathbf{comp}$ is a $s\times \log s$ matrix with its $i$-th row storing all the intermediate comparison results in the path from $x_i$ to the maximum value.

\begin{algorithm}[htbp]
\caption{Maxpool: $([y], [\mathbf{comp}]) \leftarrow \mathrm{Maxpool}([\mathbf{x}]=([x_1],\cdots,[x_s]), s), s=2^p$ for some positive integer $p$}
\label{alg:Maxpool}
\begin{algorithmic}[1]
\FOR{$i=j,\cdots, s$}
\STATE Parties set $[y_{(0,j)}]=[x_i];$
\ENDFOR
\FOR{$i = 1,\cdots,p$}
\FOR{$j = 1,\cdots,\frac{s}{2^{i}}$}
\STATE All parties compute $[c_{i, j}] = \mathrm{GEZ}([y_{(i-1,2j-1)}-y_{(i-1,2j)}]);$
\STATE Parties set $[y_{(i,j)}]=[y_{(i-1,2j)}] + [c_{i,j}]\cdot ([y_{(i-1,2j-1)}-y_{(i-1,2j)}]);$
\FOR{$k=1,\cdots, 2^{i-1}$}
\STATE Parties set $[\mathbf{comp}((j-1)\cdot 2^i + k,i)]=[c_{i,j}] + [0]$ and $[\mathbf{comp}((j-1)\cdot 2^i + 2^{i-1} + k,i)] = [1-c_{i,j}] + [0];$
\ENDFOR
\ENDFOR
\ENDFOR
\STATE Return $([y_{(p,1)}],[\mathbf{comp}]);$
\end{algorithmic}
\end{algorithm}

\begin{algorithm}[htbp]
\caption{DMaxpool: $[\mathbf{ind}] \leftarrow \mathrm{DMaxpool}([\mathbf{x}], [\mathbf{comp}])$}
\label{alg:DMaxpool}
\begin{algorithmic}[1]
\FOR{$i=1,\cdots,s$}
\STATE Parties compute $[ind_i]=\prod_{j=1}^{p} [\mathbf{comp}(i,j)];$
\ENDFOR
\STATE Return $[\mathbf{ind}]=([ind_1],\cdots,[ind_s]);$
\end{algorithmic}
\end{algorithm}

\subsection{Dropout with Derivative}
\label{dropout}

Dropout layer is performed by dropping out some values with some fixed probability $p$ such that
$$\mathrm{Dropout}(x)=
\begin{cases}
0& {\text{probability } p} \\
x/(1-p)& {\text{probability } 1-p}
\end{cases}$$

Algorithm \ref{alg:Dropout} describes our $n$-party protocol for Dropout which outputs the product of input matrix $[\mathbf{x}]$, matrix of probabilistic random bit $[\mathbf{b}]$, and public scaling factor $c=(1-p)^{-1}$ which is encoded to $\bar{c}$ using the fixed point method discussed in Section~\ref{DataRep}. Since in Step 1, the matrix of probabilistic random bits $[\mathbf{b}]$ can be generated in the offline phase of SPDZ, only one multiplication is needed for the evaluation of the Dropout layer. In addition, according to the definition of Dropout and backward propagation, the derivative of Dropout is to propagate the gradients to the nodes except for the nodes that drop their values in the forward propagation. Therefore, $\mathrm{DDropout}$ can be simply obtained from the calculation of corresponding Dropout layer, i.e., $\mathrm{DDropout}(\mathbf{[x]}) = [\mathbf{b}] \times \bar{c}$. Here $[b]$ is the matrix with random bit entries used in the corresponding Dropout layer while $\bar{c}$ is the fixed point encoding of a public scaling factor $c=(1-p)^{-1}.$

\begin{algorithm}[htbp]
\caption{Dropout: $([\mathbf{y}], [\mathbf{b}]) \leftarrow \mathrm{Dropout}([\mathbf{x}], p)$}
\label{alg:Dropout}
\begin{algorithmic}[1]
\STATE All parties call $\mathrm{PrRndBit}(p)$ to generate a matrix which has the same size of $[\mathbf{x}]$ of secret sharing probabilistic random bit $[\mathbf{b}];$
\STATE All parties compute $[\mathbf{y}] = [\mathbf{x}] \times [\mathbf{b}] \times \bar{c}$ where $\bar{c}$ is the fixed point encoding of $c=(1-p)^{-1}$ and the operator $\times$ represents the entry-wise multiplication of involved matrices;
\STATE Return $([\mathbf{y}],[\mathbf{b}]);$
\end{algorithmic}
\end{algorithm}

\section{Communication and Rounds}
\label{ComR}

We summarize the communication and round complexity of our neural network training protocols for 3PC compared with those of SecureNN \cite{wagh2019securenn} in Table \ref{table:ComRound}.

We use the same $\ell$ as the length of bits for data representation for the protocols both in SecureNN and ours. $P$ and $Q$ are the finite field size of SecureNN and our's protocol, respectively. $\mathrm{Linear}_{r, s, t}$ denotes multiplication between two matrix of dimension $r \times s$ with $s \times t$. $\mathrm{Conv2d}_{m, i, f, o}$ denotes the operations in convolutional layer with input matrix of dimension $m \times m$, $i$ input channels, $o$ output channels, and a filter of dimension $f \times f$. $\mathrm{Maxpool}_j$ and $\mathrm{DMaxpool}_j$ denote maxpool with its derivative over $j$ elements. In addition, $\mathrm{Dropout}_j$ and $\mathrm{DDropout}_j$ denote dropout with its derivative over $j$ elements, which are not available in SecureNN.

Note that (i) compared with SecureNN, we encode data in a larger finite field $\mathbb{Z}_Q$ where $Q$ is approximate $\ell+\kappa$ bits, as our design relies on SPDZ, which is more general than the specific design of SecureNN that enables protocols running over a small ring or field,

and (ii) we also do the same operation on MAC, which increases the communication cost, although Beaver triples are generated offline in our protocol thus do not need a party as "assistant", i.e., $P_2$ in \cite{wagh2019securenn}, which saves the communication rounds. We can observe the round improvement of $\mathrm{DMaxpool}$, which is because we use a general constant-round comparison protocol instead of the protocol in \cite{wagh2019securenn} consisting of share conversion, reconstruction, and multiplication. In addition, the round improvement of $\mathrm{DReLU}$ and $\mathrm{DMaxpool}$ comes from the increase of storage of intermediate comparison results, while $\mathrm{Dropout}$ and $\mathrm{DDropout}$ save the rounds by moving some steps to the offline phase of SPDZ.

\begin{table*}[ht]
\centering 
\begin{tabular}{c | c c | c c} 
\hline\hline 
 & \multicolumn{2}{c}{Rounds} & \multicolumn{2}{|c}{Communication} \\
\hline 
Protocol & SecureNN & Our's & SecureNN & Our's \\ [0.5ex] 
\hline 
$\mathrm{Linear}_{r, s, t}$ & 2 & 1 & $2(2rs + 2st + rt)\ell$ & $12(rs + st)\log Q$ \\ 
$\mathrm{Conv2d}_{m, i, f, o}$ & 2 & 1 & $2(2m^2f^2o + 2f^2oi + m^2o)\ell$ & $12(m^2f^2o + f^2oi)\log Q$ \\
$\mathrm{ReLU}$ & 2 & 3 & $10\ell$ & $12(3\ell+1) \log Q$ \\
$\mathrm{DReLU}$ & 8 & 1 & $8\ell \log P + 22\ell + 4$ & $12 \log Q$ \\
$\mathrm{Maxpool}_j$ & $9j$ & $6 \log j$ & $(8\ell \log P + 42\ell + 4)j$ & $36\ell j \log Q$ \\
$\mathrm{DMaxpool}_j$ & 2 & $\log j$ & $2(j + 1)\ell$ & $12j \log j \log Q$ \\
$\mathrm{Dropout}_j$ & NA & 1 & NA & $12j \log Q$ \\
$\mathrm{DDropout}_j$ & NA & 1 & NA & $12j \log Q$ \\ [1ex] 
\hline 
\end{tabular}
\caption{Communication and round complexity comparison for 3PC protocols} 
\label{table:ComRound} 
\end{table*}

\section{Experiments}
\label{Exp}

In this section, we present our experimental results for secure convolutional neural network training.

\textbf{System setting.} Our prototype is tested over three Linux workstations with an Intel Xeon Silver 4110 CPU (2.10GHz) and 128 GB of RAM, running CentOS7 in both LAN (in the same region) and WAN (simulated using the Linux command-line tool traffic control) settings. In the LAN setting, the average latency is 0.216 ms, and the average bandwidth is 625 MB/s, and in the WAN setting, the average latency and average bandwidth are set to be 80 ms and 100 MB/s, respectively, which were chosen to match the average network condition of three servers in Singapore, Hong Kong and Seoul in Amazon Web Services (AWS). In our experiments, data is represented in 64 bits, including 12 bits (with sign bit) for the integer part and 52 bits for the fractional part. Our protocols are implemented using Gmpy2 \cite{horsen2016gmpy2}, which is a Python version of GMP multiple-precision library and several other standard libraries.
The finite field size $Q$ is set to be a prime, which is greater than $2^{145 }$. Lastly, we set the security parameter $\kappa$ to be $80$. Note that to enable the comparison between SecureNN and our protocols, our experiments use the same bit length to represent the data $\ell$ as the one used in the experiment conducted in SecureNN\cite{wagh2019securenn}.

\textbf{Neural network architecture.} We implement two types of neural network: a deep neural network and a convolutional neural network. The former is the same model as used in \cite{wagh2019securenn} and \cite{mohassel2017secureml}, with architecture of fully connected layer (784, 128) - ReLU - fully connected layer (128, 10)- ReLU. The latter has the architecture of padding (32, 32) - convolutional layer (4, 28, 28) - ReLU - Maxpool (4, 14, 14) - convolutional layer (12, 10, 10) - ReLU - dropout - Maxpool (12, 5, 5) - flatten (1, 300) - fully connected layer (300, 120) - fully connected layer (120, 10) - ReLU. Both neural networks are implemented based on reproduced SecureNN protocols in \cite{wagh2019securenn} and our protocols, while plaintext implementation is based on Numpy \cite{oliphant2006guide}. We use MNIST dataset \cite{lecun-mnisthandwrittendigit-2010} which consists of 70,000 black-white hand-written digit images of size $28 \times 28$ in 10 classes. In our experiment, 60,000 images are used for training and 10,000 images are used for testing. Note that we only give the performance evaluation on neural network training, i.e., the online phase of SPDZ, as "raw materials" to be used can be prepared in the offline phase.

\begin{table}[!htbp]
\centering
\begin{tabular}{cccc}
\hline
Type & Epochs & Accuracy & Training time (LAN\slash WAN)\\
\hline
\multirow{3}*{DNN} & 1 & 95.03\% & 0.29\slash 3.99 hr\\
& 5 & 96.99\% & 1.45\slash 19.98 hr\\
& 10 & 97.75\% & 2.92\slash 40.25 hr\\
\hline
\multirow{3}*{CNN} & 1 & 97.00\% & 1.65\slash 8.51 hr\\
& 5 & 97.94\% & 8.27\slash 42.67 hr\\
& 10 & 98.08\% & 16.50\slash 85.28 hr\\
\hline
\end{tabular}
\caption{Secure training with different epochs for 3PC in LAN\slash WAN setting with batch size 64}
\label{BatchTraining}
\end{table}

\begin{table}[!htbp]
\centering
\begin{tabular}{cccc}
\hline
Type & Batch size & Accuracy & Training time (LAN\slash WAN)\\
\hline
\multirow{3}*{DNN} & 16 & 94.99\% & 0.36 \slash 5.72 hr\\
& 64 & 95.03\% & 0.29 \slash 3.99 hr\\
& 128 & 96.75\% & 0.18 \slash 2.26 hr\\
\hline
\multirow{3}*{CNN} & 16 & 96.11\% & 1.94 \slash 11.52 hr\\
& 64 & 97.00\% & 1.65 \slash 8.51 hr\\
& 128 & 97.05\% & 1.42 \slash 6.69 hr\\
\hline
\end{tabular}
\caption{Secure training with different batch size for 3PC in LAN\slash WAN setting for 1 epoch}
\label{EpochTraining}
\end{table}

\begin{table}[htbp]
\centering
\begin{tabular}{ccccc}
\hline
Type & Protocol & Accuracy & Time (LAN\slash WAN) & Comm\\
\hline
\multirow{3}*{DNN} & SecureNN & 94.21\% & 0.13\slash 1.48 hr & 6.08 MB\\
& Our's & 94.03\% & 0.29\slash 3.99 hr & 460.82 MB\\
& Plaintext & 95.92\% & 10.47 s & NA\\
\hline
\multirow{3}*{CNN} & SecureNN & 97.01\% & 0.78\slash 3.54 hr & 56.92 MB\\
& Our's & 97.00\% & 1.65 \slash 8.51 hr & 2.42 GB\\
& Plaintext & 97.07\% & 53.47 s & NA\\
\hline
\end{tabular}
\caption{Training performance comparison with SecureNN for 1 epoch in LAN\slash WAN setting}
\label{TrainingComparison}
\end{table}

As shown in Table \ref{BatchTraining}, with a batch size of 64, our protocol offers a prediction accuracy of 97.75\% after ten epochs for DNN, which takes 2.92 hours in the LAN setting. For CNN, it takes 16.50 hours in LAN setting to complete ten epochs training and achieves an accuracy of 98.08\%. Table \ref{EpochTraining} shows the training time of 1 epoch with different batch sizes for DNN and CNN. Table \ref{TrainingComparison} summarizes the training comparison between our protocol, SecureNN, and plaintext in terms of training time and communication cost. The results show that the ratio of training time increases vastly from DNN to CNN regarding different network settings as our scheme involves higher communication costs. However, we can observe that our protocol improves the threat model from semi-honest to dishonest majority with affordable overheads of around 2.1X and 2.7X in LAN and WAN settings, respectively.

\section{Conclusions}
\label{Conclu}

In this paper, we propose a new scheme with several primitives for secure neural network training in malicious majority setting leveraging on SPDZ. Our experimental results show that our protocols offer active security with affordable overheads of around 2X and 2.5X in LAN and WAN time, respectively, compared with existing schemes in the semi-honest setting. Besides, we propose a scheme for Beaver triple conversion from a ring $\mathbb{Z}_N$ to a finite field $\mathbb{Z}_Q$ to enable MAC checking in SPDZ, relying on two instances of Paillier crypto-systems.

\bibliographystyle{ACM-Reference-Format}
\bibliography{sample-base}


\begin{thebibliography}{21}


\ifx \showCODEN    \undefined \def \showCODEN     #1{\unskip}     \fi
\ifx \showDOI      \undefined \def \showDOI       #1{#1}\fi
\ifx \showISBNx    \undefined \def \showISBNx     #1{\unskip}     \fi
\ifx \showISBNxiii \undefined \def \showISBNxiii  #1{\unskip}     \fi
\ifx \showISSN     \undefined \def \showISSN      #1{\unskip}     \fi
\ifx \showLCCN     \undefined \def \showLCCN      #1{\unskip}     \fi
\ifx \shownote     \undefined \def \shownote      #1{#1}          \fi
\ifx \showarticletitle \undefined \def \showarticletitle #1{#1}   \fi
\ifx \showURL      \undefined \def \showURL       {\relax}        \fi
\providecommand\bibfield[2]{#2}
\providecommand\bibinfo[2]{#2}
\providecommand\natexlab[1]{#1}
\providecommand\showeprint[2][]{arXiv:#2}

\bibitem[\protect\citeauthoryear{??}{sec}{2009}]%
        {securescm}
 \bibinfo{year}{2009}\natexlab{}.
\newblock \showarticletitle{SecureSCM}. In
  \bibinfo{booktitle}{\emph{Deliverable D9.2, EU FP7 Project Secure Supply
  Chain Management (SecureSCM)}}.
\newblock


\bibitem[\protect\citeauthoryear{Bendlin, Damg{\aa}rd, Orlandi, and
  Zakarias}{Bendlin et~al\mbox{.}}{2011}]%
        {bendlin2011semi}
\bibfield{author}{\bibinfo{person}{Rikke Bendlin}, \bibinfo{person}{Ivan
  Damg{\aa}rd}, \bibinfo{person}{Claudio Orlandi}, {and} \bibinfo{person}{Sarah
  Zakarias}.} \bibinfo{year}{2011}\natexlab{}.
\newblock \showarticletitle{Semi-homomorphic encryption and multiparty
  computation}. In \bibinfo{booktitle}{\emph{Annual International Conference on
  the Theory and Applications of Cryptographic Techniques}}. Springer,
  \bibinfo{pages}{169--188}.
\newblock


\bibitem[\protect\citeauthoryear{Canetti}{Canetti}{2001}]%
        {canetti2001universally}
\bibfield{author}{\bibinfo{person}{Ran Canetti}.}
  \bibinfo{year}{2001}\natexlab{}.
\newblock \showarticletitle{Universally composable security: A new paradigm for
  cryptographic protocols}. In \bibinfo{booktitle}{\emph{Proceedings 42nd IEEE
  Symposium on Foundations of Computer Science}}. IEEE,
  \bibinfo{pages}{136--145}.
\newblock


\bibitem[\protect\citeauthoryear{Catrina and De~Hoogh}{Catrina and
  De~Hoogh}{2010}]%
        {catrina2010improved}
\bibfield{author}{\bibinfo{person}{Octavian Catrina} {and}
  \bibinfo{person}{Sebastiaan De~Hoogh}.} \bibinfo{year}{2010}\natexlab{}.
\newblock \showarticletitle{Improved primitives for secure multiparty integer
  computation}. In \bibinfo{booktitle}{\emph{International Conference on
  Security and Cryptography for Networks}}. Springer,
  \bibinfo{pages}{182--199}.
\newblock


\bibitem[\protect\citeauthoryear{Catrina and Saxena}{Catrina and
  Saxena}{2010}]%
        {catrina2010secure}
\bibfield{author}{\bibinfo{person}{Octavian Catrina} {and}
  \bibinfo{person}{Amitabh Saxena}.} \bibinfo{year}{2010}\natexlab{}.
\newblock \showarticletitle{Secure computation with fixed-point numbers}. In
  \bibinfo{booktitle}{\emph{International Conference on Financial Cryptography
  and Data Security}}. Springer, \bibinfo{pages}{35--50}.
\newblock


\bibitem[\protect\citeauthoryear{Cramer, Damg{\aa}rd, and Nielsen}{Cramer
  et~al\mbox{.}}{2015}]%
        {cramer2015secure}
\bibfield{author}{\bibinfo{person}{Ronald Cramer}, \bibinfo{person}{Ivan~Bjerre
  Damg{\aa}rd}, {and} \bibinfo{person}{Jesper~Buus Nielsen}.}
  \bibinfo{year}{2015}\natexlab{}.
\newblock \bibinfo{booktitle}{\emph{Secure multiparty computation}}.
\newblock \bibinfo{publisher}{Cambridge University Press}.
\newblock


\bibitem[\protect\citeauthoryear{Cramer, Damgrd, Escudero, Scholl, and
  Xing}{Cramer et~al\mbox{.}}{2018}]%
        {cramerspdz2k}
\bibfield{author}{\bibinfo{person}{R Cramer}, \bibinfo{person}{I Damgrd},
  \bibinfo{person}{D Escudero}, \bibinfo{person}{P Scholl}, {and}
  \bibinfo{person}{C Xing}.} \bibinfo{year}{2018}\natexlab{}.
\newblock \showarticletitle{SPDZ2k: efficient MPC mod $2^k$ for dishonest
  majority}. In \bibinfo{booktitle}{\emph{Annual international cryptology
  conference}}.
\newblock


\bibitem[\protect\citeauthoryear{Damg{\aa}rd, Fitzi, Kiltz, Nielsen, and
  Toft}{Damg{\aa}rd et~al\mbox{.}}{2006}]%
        {damgaard2006unconditionally}
\bibfield{author}{\bibinfo{person}{Ivan Damg{\aa}rd}, \bibinfo{person}{Matthias
  Fitzi}, \bibinfo{person}{Eike Kiltz}, \bibinfo{person}{Jesper~Buus Nielsen},
  {and} \bibinfo{person}{Tomas Toft}.} \bibinfo{year}{2006}\natexlab{}.
\newblock \showarticletitle{Unconditionally secure constant-rounds multi-party
  computation for equality, comparison, bits and exponentiation}. In
  \bibinfo{booktitle}{\emph{Theory of Cryptography Conference}}. Springer,
  \bibinfo{pages}{285--304}.
\newblock


\bibitem[\protect\citeauthoryear{Damg{\aa}rd, Pastro, Smart, and
  Zakarias}{Damg{\aa}rd et~al\mbox{.}}{2012}]%
        {damgaard2012multiparty}
\bibfield{author}{\bibinfo{person}{Ivan Damg{\aa}rd}, \bibinfo{person}{Valerio
  Pastro}, \bibinfo{person}{Nigel Smart}, {and} \bibinfo{person}{Sarah
  Zakarias}.} \bibinfo{year}{2012}\natexlab{}.
\newblock \showarticletitle{Multiparty computation from somewhat homomorphic
  encryption}. In \bibinfo{booktitle}{\emph{Annual Cryptology Conference}}.
  Springer, \bibinfo{pages}{643--662}.
\newblock


\bibitem[\protect\citeauthoryear{Demmler, Schneider, and Zohner}{Demmler
  et~al\mbox{.}}{2015}]%
        {demmler2015aby}
\bibfield{author}{\bibinfo{person}{Daniel Demmler}, \bibinfo{person}{Thomas
  Schneider}, {and} \bibinfo{person}{Michael Zohner}.}
  \bibinfo{year}{2015}\natexlab{}.
\newblock \showarticletitle{ABY-A framework for efficient mixed-protocol secure
  two-party computation.}. In \bibinfo{booktitle}{\emph{NDSS}}.
\newblock


\bibitem[\protect\citeauthoryear{Hazay, Mikkelsen, Rabin, Toft, and
  Nicolosi}{Hazay et~al\mbox{.}}{2019}]%
        {hazay2019efficient}
\bibfield{author}{\bibinfo{person}{Carmit Hazay},
  \bibinfo{person}{Gert~L{\ae}ss{\o}e Mikkelsen}, \bibinfo{person}{Tal Rabin},
  \bibinfo{person}{Tomas Toft}, {and} \bibinfo{person}{Angelo~Agatino
  Nicolosi}.} \bibinfo{year}{2019}\natexlab{}.
\newblock \showarticletitle{Efficient RSA key generation and threshold paillier
  in the two-party setting}.
\newblock \bibinfo{journal}{\emph{Journal of Cryptology}} \bibinfo{volume}{32},
  \bibinfo{number}{2} (\bibinfo{year}{2019}), \bibinfo{pages}{265--323}.
\newblock


\bibitem[\protect\citeauthoryear{Horsen}{Horsen}{2016}]%
        {horsen2016gmpy2}
\bibfield{author}{\bibinfo{person}{CV Horsen}.}
  \bibinfo{year}{2016}\natexlab{}.
\newblock \bibinfo{title}{Gmpy2: Mupltiple-precision arithmetic for python}.
\newblock
\newblock


\bibitem[\protect\citeauthoryear{Juvekar, Vaikuntanathan, and
  Chandrakasan}{Juvekar et~al\mbox{.}}{2018}]%
        {juvekar2018gazelle}
\bibfield{author}{\bibinfo{person}{Chiraag Juvekar}, \bibinfo{person}{Vinod
  Vaikuntanathan}, {and} \bibinfo{person}{Anantha Chandrakasan}.}
  \bibinfo{year}{2018}\natexlab{}.
\newblock \showarticletitle{$\{$GAZELLE$\}$: A low latency framework for secure
  neural network inference}. In \bibinfo{booktitle}{\emph{27th $\{$USENIX$\}$
  Security Symposium ($\{$USENIX$\}$ Security 18)}}.
  \bibinfo{pages}{1651--1669}.
\newblock


\bibitem[\protect\citeauthoryear{Kairouz, McMahan, Avent, Bellet, Bennis,
  Bhagoji, Bonawitz, Charles, Cormode, Cummings, et~al\mbox{.}}{Kairouz
  et~al\mbox{.}}{2019}]%
        {kairouz2019advances}
\bibfield{author}{\bibinfo{person}{Peter Kairouz}, \bibinfo{person}{H~Brendan
  McMahan}, \bibinfo{person}{Brendan Avent}, \bibinfo{person}{Aur{\'e}lien
  Bellet}, \bibinfo{person}{Mehdi Bennis}, \bibinfo{person}{Arjun~Nitin
  Bhagoji}, \bibinfo{person}{Keith Bonawitz}, \bibinfo{person}{Zachary
  Charles}, \bibinfo{person}{Graham Cormode}, \bibinfo{person}{Rachel
  Cummings}, {et~al\mbox{.}}} \bibinfo{year}{2019}\natexlab{}.
\newblock \showarticletitle{Advances and open problems in federated learning}.
\newblock \bibinfo{journal}{\emph{arXiv preprint arXiv:1912.04977}}
  (\bibinfo{year}{2019}).
\newblock


\bibitem[\protect\citeauthoryear{Keller, Orsini, and Scholl}{Keller
  et~al\mbox{.}}{2016}]%
        {keller2016mascot}
\bibfield{author}{\bibinfo{person}{Marcel Keller}, \bibinfo{person}{Emmanuela
  Orsini}, {and} \bibinfo{person}{Peter Scholl}.}
  \bibinfo{year}{2016}\natexlab{}.
\newblock \showarticletitle{MASCOT: faster malicious arithmetic secure
  computation with oblivious transfer}. In
  \bibinfo{booktitle}{\emph{Proceedings of the 2016 ACM SIGSAC Conference on
  Computer and Communications Security}}. \bibinfo{pages}{830--842}.
\newblock


\bibitem[\protect\citeauthoryear{LeCun and Cortes}{LeCun and Cortes}{2010}]%
        {lecun-mnisthandwrittendigit-2010}
\bibfield{author}{\bibinfo{person}{Yann LeCun} {and} \bibinfo{person}{Corinna
  Cortes}.} \bibinfo{year}{2010}\natexlab{}.
\newblock \showarticletitle{{MNIST} handwritten digit database}.
\newblock \bibinfo{howpublished}{http://yann.lecun.com/exdb/mnist/}.
\newblock  (\bibinfo{year}{2010}).
\newblock
\urldef\tempurl%
\url{http://yann.lecun.com/exdb/mnist/}
\showURL{%
\tempurl}


\bibitem[\protect\citeauthoryear{Liu, Juuti, Lu, and Asokan}{Liu
  et~al\mbox{.}}{2017}]%
        {liu2017oblivious}
\bibfield{author}{\bibinfo{person}{Jian Liu}, \bibinfo{person}{Mika Juuti},
  \bibinfo{person}{Yao Lu}, {and} \bibinfo{person}{Nadarajah Asokan}.}
  \bibinfo{year}{2017}\natexlab{}.
\newblock \showarticletitle{Oblivious neural network predictions via minionn
  transformations}. In \bibinfo{booktitle}{\emph{Proceedings of the 2017 ACM
  SIGSAC Conference on Computer and Communications Security}}.
  \bibinfo{pages}{619--631}.
\newblock


\bibitem[\protect\citeauthoryear{Mohassel and Zhang}{Mohassel and
  Zhang}{2017}]%
        {mohassel2017secureml}
\bibfield{author}{\bibinfo{person}{Payman Mohassel} {and}
  \bibinfo{person}{Yupeng Zhang}.} \bibinfo{year}{2017}\natexlab{}.
\newblock \showarticletitle{Secureml: A system for scalable privacy-preserving
  machine learning}. In \bibinfo{booktitle}{\emph{2017 IEEE Symposium on
  Security and Privacy (SP)}}. IEEE, \bibinfo{pages}{19--38}.
\newblock


\bibitem[\protect\citeauthoryear{Oliphant}{Oliphant}{2006}]%
        {oliphant2006guide}
\bibfield{author}{\bibinfo{person}{Travis~E Oliphant}.}
  \bibinfo{year}{2006}\natexlab{}.
\newblock \bibinfo{booktitle}{\emph{A guide to NumPy}}.
  Vol.~\bibinfo{volume}{1}.
\newblock \bibinfo{publisher}{Trelgol Publishing USA}.
\newblock


\bibitem[\protect\citeauthoryear{Paillier}{Paillier}{1999}]%
        {paillier1999public}
\bibfield{author}{\bibinfo{person}{Pascal Paillier}.}
  \bibinfo{year}{1999}\natexlab{}.
\newblock \showarticletitle{Public-key cryptosystems based on composite degree
  residuosity classes}. In \bibinfo{booktitle}{\emph{International conference
  on the theory and applications of cryptographic techniques}}. Springer,
  \bibinfo{pages}{223--238}.
\newblock


\bibitem[\protect\citeauthoryear{Wagh, Gupta, and Chandran}{Wagh
  et~al\mbox{.}}{2019}]%
        {wagh2019securenn}
\bibfield{author}{\bibinfo{person}{Sameer Wagh}, \bibinfo{person}{Divya Gupta},
  {and} \bibinfo{person}{Nishanth Chandran}.} \bibinfo{year}{2019}\natexlab{}.
\newblock \showarticletitle{Securenn: 3-party secure computation for neural
  network training}.
\newblock \bibinfo{journal}{\emph{Proceedings on Privacy Enhancing
  Technologies}} \bibinfo{volume}{2019}, \bibinfo{number}{3}
  (\bibinfo{year}{2019}), \bibinfo{pages}{26--49}.
\newblock


\end{thebibliography}

\appendix
\section{Functionalities}\label{app:funct}
In this section, we provide the ideal functionalities of the protocols that we propose in Section~\ref{ProposedProtocols}.
\subsection{Comparison Modulo $N$}\label{subapp:compfunct}
The main protocol we are considering here is $\mathcal{F}_{GEZ,N}$ which receives a secretly shared value $[x]^N$ and the bit length $k$ of $x$ as input and outputs $[s]^N$ where $s=(x\geq 0).$ The full specification of the functionality can be found in Functionality~\ref{funct:GEZ}.

\begin{funct}[htbp]
\caption{Comparison Functionality: $([s]^N) \leftarrow \mathcal{F}_{GEZ}([\mathbf{x}]^N, k)$}
\label{funct:GEZ}
\begin{algorithmic}[1]
\STATE Receive the shares of $[\mathbf{x}]^N$ from the parties along with the input $k;$
\STATE Reconstruct $\mathbf{x}$ and calculate $s=(x\geq 0),$ i.e. $s=1$ if $x\geq 0$ and $0$ otherwise;
\STATE Generate a secret sharing of $s, [s]^N;$
\STATE Send $s_i$ to $P_i$ for all $i; $
\end{algorithmic}
\end{funct}

In the construction of a protocol securely realizing $\mathcal{F}_{GEZ,N},$ we require several other supporting functionalities. The first supporting functionality we will consider is $\mathcal{F}_{BitLTC}$ which receives a public $k$-bit integer $x$ and the secretly shared bit decomposition of another $k$-bit value $[y_{k-1}]^N,\cdots, [y_0]^N$ where  $y=\sum_{i=0}^{k-1} y_i 2^i.$ The functionality outputs $[s]^N$ such that $s=(x<y).$ The full specification of the functionality can be found in Functionality~\ref{funct:BitLTC}.  
\begin{funct}[htbp]
\caption{Bit Comparison: $([s]^N) \leftarrow \mathcal{F}_{BitLTC}(x,([y_{k-1}]^N,\cdots,[y_0]^N))$}
\label{funct:BitLTC}
\begin{algorithmic}[1]
\STATE Receive the shares of $[y_i]^N$ from the parties for $i=0,\cdots, k-1;$
\STATE Reconstruct $\mathbf{y}$ and calculate $s=(x<y),$ i.e. $s=1$ if $x<y$ and $0$ otherwise;
\STATE Generate a secret sharing of $s, [s]^N;$
\STATE Send $s_i$ to $P_i$ for all $i; $
\end{algorithmic}
\end{funct}

 In our work, the protocol securely realizing $\mathcal{F}_{BitLTC}$ exactly follows the one presented in~\cite{securescm}. The protocol presented in~\cite{securescm} depends on two other functionalities, $\mathcal{F}_{RndBit},$ which produces a random shared bit, and $\mathcal{F}_{RndInv},$ which produces a random invertible element. For completeness, the full specifications of $\mathcal{F}_{RndBit}$ and $\mathcal{F}_{RndInv}$ can be found in Functionalities~\ref{funct:RndBit} and~\ref{funct:RndInv} respectively.
 
 \begin{funct}[htbp]
\caption{Random Bit Generation: $([a]^N) \leftarrow \mathcal{F}_{RndBit}()$}
\label{funct:RndBit}
\begin{algorithmic}[1]
\STATE Generate a random bit $a;$
\STATE generate a secret sharing of $[a]^N;$
\STATE Send $a_i$ to $P_i$ for all $i;$
\end{algorithmic}
\end{funct}

  \begin{funct}[htbp]
\caption{Random Invertible Element Generation: $([r]^N,[r^{-1}]^N) \leftarrow \mathcal{F}_{RndInv}(N)$}
\label{funct:RndInv}
\begin{algorithmic}[1]
\STATE Generate a random invertible element $r$ and calculate its inverse $r^{-1};$
\STATE Generate a secret sharing of $[r]^N$ and $[r^{-1}]^N;$
\STATE Send $r_i$ and $r_i^{-1}$ to $P_i$ for all $i;$
\end{algorithmic}
\end{funct}
 
 The correctness and security of $\mathrm{BitLTC}$ can be easily verified and can be found in Section~\ref{subapp:compcorsec}.
 
 In contrast to our situation, the protocols securely realizing the two functionalities that are proposed in~\cite{securescm} only works when the underlying ring is a field. Since we want the protocol to be applicable over $\mathbb{Z}_N$ for some RSA modulus, we complete Subsection~\ref{subsec:comp} by proposing new protocols securely realizing both $\mathcal{F}_{RndBit}$ and $\mathcal{F}_{RndInv}$ over $\mathbb{Z}_N.$ 
 
 \subsection{Wrap, Modulo Reduction, and Share Conversion}\label{subapp:MSWfunct}
 First we discuss the first supporting functionality, $\mathcal{F}_{LiftWrap}.$ Intuitively, $\mathcal{F}_{LiftWrap}$ takes a secretly shared $[x]^N$ modulo an RSA modulus $N$ and another RSA modulus $N'$ as inputs and outputs $[\delta]^{N'}$ where $\delta$ is the number of wraps needed in the calculation of $x_1+\cdots+x_n\pmod{N}.$ The full specification of the functionality can be found in Functionality~\ref{funct:LiftWrap}.

 \begin{funct}[htbp]
\caption{Lift Wrap functionality: $([\delta]^{N'}) \leftarrow \mathcal{F}_{LiftWrap}([x]^N, N')$}
\label{funct:LiftWrap}
\begin{algorithmic}[1]
\STATE Receive the shares $x_i$ of $[x]^N$ from all parties and reconstruct $\bar{x}=x_1+x_2+\cdots+x_n;$
\STATE Compute $\delta=\left\lfloor \frac{\bar{x}}{N}\right\rfloor;$
\STATE Generate a secret sharing of $[\delta]^{N'};$
\STATE Send $\delta_i$ to $P_i$ for all $i;$
\end{algorithmic}
\end{funct}

Next we discuss the next supporting functionality, $\mathcal{F}_{LiftMod}$ which takes a secretly shared $[x]^N$ modulo an RSA modulus $N$ and another RSA modulus $N'$ as inputs and outputs $[x]^{N'},$ which is the secret sharing of the same value $x$ modulo $N'.$ The full specification of $\mathcal{F}_{LiftMod}$ can be found in Functionality~\ref{funct:LiftMod}.

 \begin{funct}[htbp]
\caption{Lift Mod functionality: $([x]^{N'}) \leftarrow \mathcal{F}_{LiftMod}([x]^N, N')$}
\label{funct:LiftMod}
\begin{algorithmic}[1]
\STATE Receive the shares $x_i$ of $[x]^N$ from all parties, reconstruct $x\equiv x_1+x_2+\cdots+x_n\pmod{N}$ and treat $x$ as an element of $\mathbb{Z}_{N'};$
\STATE Generate a secret sharing of $[x]^{N'};$
\STATE Send $x_i^{N'}$ to $P_i$ for all $i;$
\end{algorithmic}
\end{funct}

We proceed to the next functionality, $\mathcal{F}_{DropMod}$ which takes a secretly shared $[x]^{N'}$ modulo an RSA modulus $N'$ and another RSA modulus $N$ as inputs given that $x<N$ and outputs $[x]^{N},$ which is the secret sharing of the same value $x$ modulo $N.$ The full specification of $\mathcal{F}_{DropMod}$ can be found in Functionality~\ref{funct:DropMod}.

 \begin{funct}[htbp]
\caption{Drop Mod functionality: $([x]^{N}) \leftarrow \mathcal{F}_{DropMod}([x]^N, N')$}
\label{funct:DropMod}
\begin{algorithmic}[1]
\STATE Receive the shares $x_i$ of $[x]^{N'}$ from all parties, reconstruct $x\equiv x_1+x_2+\cdots+x_n\pmod{N'}$ and treat $x$ as an element of $\mathbb{Z}_{N};$
\STATE Generate a secret sharing of $[x]^{N};$
\STATE Send $x_i^{N}$ to $P_i$ for all $i;$
\end{algorithmic}
\end{funct}

A direct application of $\mathcal{F}_{DropMod}$ is the functionality $\mathcal{F}_{Wrap},$ which has $[x]^N$ and $N'$ as inputs and $[\delta]^N$ as output where $\delta$ is as defined in the discussion of $\mathcal{F}_{LiftWrap}.$ The full specification of $\mathcal{F}_{Wrap}$ can be found in Functionality~\ref{funct:Wrap}.

 \begin{funct}[htbp]
\caption{Wrap functionality: $([\delta]^{N}) \leftarrow \mathcal{F}_{Wrap}([x]^N)$}
\label{funct:Wrap}
\begin{algorithmic}[1]
\STATE Receive the shares $x_i$ of $[x]^N$ from all parties and reconstruct $\bar{x}=x_1+x_2+\cdots+x_n;$
\STATE Compute $\delta=\left\lfloor \frac{\bar{x}}{N}\right\rfloor;$
\STATE Generate a secret sharing of $[\delta]^{N};$
\STATE Send $\delta_i$ to $P_i$ for all $i;$
\end{algorithmic}
\end{funct}

Lastly, we discuss the main functionality we want to achieve in this section, $\mathcal{F}_{ShaConv},$ which converts a secretly shared value $[x]^N$ to $[x]^Q$ for the inputted integer $Q.$ The full specification of $\mathcal{F}_{ShaConv}$ can be found in Functionality~\ref{funct:ShaConv}.

\begin{funct}[htbp]
\caption{Share Conversion functionality: $([x]^{Q}) \leftarrow \mathcal{F}_{ShaConv}([x]^N,Q)$}
\label{funct:ShaConv}
\begin{algorithmic}[1]
\STATE Receive the shares $x_i$ of $[x]^{N}$ from all parties, reconstruct $x\equiv x_1+x_2+\cdots+x_n\pmod{N};$
\STATE Generate a secret sharing of $[x]^{Q};$
\STATE Send $x_i^{Q}$ to $P_i$ for all $i;$
\end{algorithmic}
\end{funct}

We note that the functionalities $\mathcal{F}_{LiftMod},\mathcal{F}_{DropMod},$ and $\mathcal{F}_{ShaConv}$ are exactly the same except for the source and destination modulo. However, to have a clear distinction between the protocols realizing them, we use different functionalities for the different uses. 

\subsection{Beaver Triple Conversion}\label{subapp:BTCfunct}
In this section, we provide the functionality, $\mathcal{F}_{TripConv}.$ Intuitively, $\mathcal{F}_{TripConv}$ takes a secretly shared Beaver Triple $([a]^N,[b]^N,[c]^N)$ over an RSA modulus $N$ as well as a prime $Q.$ It then outputs the a Beaver Triple defined over $Q$ such that the first two elements are $[a]^Q$ and $[b]^Q.$ The full specification of the functionality can be found in Functionality~\ref{funct:BTCfunct}.

\begin{funct}[htbp]
\caption{Beaver Triple Conversion functionality: $([a]^{Q},[b]^Q,[ab\pmod{Q}]^Q) \leftarrow \mathcal{F}_{TripConv}(([a]^N,[b]^N,[ab\pmod{N}]^N),Q)$}
\label{funct:BTCfunct}
\begin{algorithmic}[1]
\STATE Receive the shares $x_i$ of $[x]^{N}$ from all parties, reconstruct $x\equiv x_1+x_2+\cdots+x_n\pmod{N};$
\STATE Generate a secret sharing of $[x]^{Q};$
\STATE Send $x_i^{Q}$ to $P_i$ for all $i;$
\end{algorithmic}
\end{funct}
\subsection{Probabilistic Bit Generation}\label{subapp:PBGfunct}
In this section, we provide the functionality, $\mathcal{F}_{PrRndBit}.$ The functionality $\mathcal{F}_{PrRndBit}$ takes a prime $Q$ and a probability $p\in\{0,1\}$ and outputs a secretly shared $[b]^Q$ where $b=1$ with probability $p^\ast$ and $b=0$ with probability $1-p^\ast$ where $p^\ast=\frac{\lfloor pQ\rfloor}{Q}.$ We note that $p^\ast\approx p$ when $Q$ is sufficiently large. The full specification of the functionality can be found in Functionality~\ref{funct:PrRndBitfunct}.

\begin{funct}[htbp]
\caption{Probabilistic Bit Generation functionality: $([b]^Q) \leftarrow \mathcal{F}_{PrRndBit}(Q,p)$}
\label{funct:PrRndBitfunct}
\begin{algorithmic}[1]
\STATE Sample $b$ from a Bernoulli random variable with success probability $p^\ast;$
\STATE Generate a secret sharing of $[b]^{Q};$
\STATE Send $b_i^{Q}$ to $P_i$ for all $i;$
\end{algorithmic}
\end{funct}

\section{Correctness and Security Claims}\label{app:corsec}
In this section, we provide the correctness and security claims along with their proofs for protocols proposed in Section~\ref{ProposedProtocols}.
\subsection{Comparison Modulo $N$}\label{subapp:compcorsec}
First we provide the security claims for $\mathrm{RndBit}$ and $\mathrm{RndInv}.$
\begin{proposition}\label{prop:RndBit}
The output $[a]^N$ of $\mathrm{RndBit}$ satisfies $a\in\{0,1\}$ except with probability $\frac{2}{N-2}.$ Furthermore, 
The protocol $\mathrm{RndBit}$ correctly and securely realizes $\mathcal{F}_{RndBit}$ against computational adversary controlling all but one parties. 
\end{proposition}
\begin{proof}
Assuming that $a_i^{(0)}\in\{0,1\}$ for all $i=1,\cdots, n,$ it is easy to see that Lines $2$ up to $6$ of Algorithm~\ref{alg:RndBit} calculates the XOR of the $n$ bits. Hence, we have $a\in\{0,1\}$ as required. Note that since we are considering active adversary, it is possible for $a_i\notin\{0,1\}$ for corrupted $P_i.$ Instead of checking each $a_i$ independently, we are only checking that the final result $a.$ We perform the check by checking whether $a(1-a)\equiv 0\pmod{N}.$ Recall that $N=pq$ for a pair of distinct primes $p$ and $q.$ Note that there are $4$ solutions of the equation $a-a^2\equiv 0\pmod{N}$, namely, $0,1,x,$ and $y\in \mathbb{Z}_N$ where $x$ is the unique solution of the system of equations $x\equiv 0 \pmod{p}$ and $x\equiv 1\pmod{q}$ while $y$ is the unique solution of the system of the equations $y\equiv 1\pmod{p}$ and $y\equiv 0\pmod{q}.$ Note that by the Chinese Remainder Theorem, the problem of finding $x$ and $y$ is equivalent to the problem of factoring $N$ to $p$ and $q.$ Hence, assuming the security of Paillier cryptosystem, we can assume that the adversary does not have the access to the value of $x$ or $y.$ Combined with the fact that $a=\bigoplus_{i=1}^n a_i$ where there exists at least one $a_i\in\{0,1\}$ that is uniformly distributed, unless $a\in\{0,1\},$ we have $a(1-a)=0$ with probability $\frac{2}{n-2}.$ 

Now we consider the security. Note that assuming the security of the multiplication protocol, the only extra value that is revealed is the value $check.$ Since we are assuming that the adversary does not have any access to the non-trivial idempotent elements $x$ and $y,$ unless $a_i\in\{0,1\},$ the protocol aborts with high probability for sufficiently large $N.$ Hence with high probability, $a=\bigoplus_{i=1}^n a_i\in\{0,1\}.$ Since there exists at least one $i$ such that $a_i$ is uniformly distributed, we have $a$ to be uniformly distributed among $\{0,1\},$ the same as the output of $\mathcal{F}_{RndBit}.$ 
\end{proof}

\begin{proposition}\label{prop:RndInv}
The protocol $\mathrm{RndInv}$ correctly and securely realizes $\mathcal{F}_{RndInv}$ against computational adversary controlling all but one parties.
\end{proposition}

\begin{proof}
Note that if $u$ is invertible, then $x$ and $y$ are also invertible. Furthermore, note that for any invertible elements $x,y\in \mathbb{Z}^N,$ there exists an invertible $z\in \mathbb{Z}^N$ such that $x=zy.$ Hence the knowledge of $u$ does not change the distribution of $x.$ This proves that the outputted invertible element $r$ is uniformly distributed among all possible invertible elements in $\mathbb{Z}^n.$

Now we show that $\mathrm{RndInv}$ securely realizes $\mathcal{F}_{RndInv}$ under the adversary assumption claimed. Note that assuming the security of the multiplication protocol between two shared values modulo $N,$ the only extra information the adversary learns from $\mathrm{RndInv}$ compared to $\mathcal{F}_{RndInv}$ is the value of $u.$ However, as observed in the correctness argument, the value of $u$ is independent of the value of $x,$ which shows that such information does not provide information about $[x]^N,$ proving the security of $\mathrm{RndInv}.$
\end{proof}

We note that due to its exact same form as the protocol $\mathrm{BitLTC}$ constructed over fields in~\cite{securescm}, under the $(\mathcal{F}_{RndBit},\mathcal{F}_{RndInv})$ - hybrid model, we also have the following correctness and security result for $\mathrm{BitLTC}$ over $\mathbb{Z}_N.$
\begin{proposition}
The protocol $\mathrm{BitLTC}$ correctly and securely realizes $\mathcal{F}_{BitLTC}$ in integer rings modulo RSA modulus $\mathbb{Z}_N$ in the $(\mathcal{F}_{RndBit},\mathcal{F}_{RndInv})$-hybrid model. 
 \end{proposition}
 
Lastly, we consider $\mathrm{GEZ}_N.$ Similar to $\mathrm{BitLTC},$ since the protocol $\mathrm{GEZ}_N$ closely follow the protocol with the same name that is proposed in~\cite{securescm}, under $(\mathcal{F}_{RndBit},\mathcal{F}_{PRndInt},\mathcal{F}_{BitLTC})$-hybrid model, we also have the same correctness and security result for $\mathrm{GEZ}_N$
\begin{proposition}
The protocol $\mathrm{GEZ}_N$ correctly and securely realizes $\mathcal{F}_{GEZ}$ in integer rings modulo RSA modulus $\mathbb{Z}_N$ in the $(\mathcal{F}_{RndBit},$ $\mathcal{F}_{PRndInt},\mathcal{F}_{BitLTC})$-hybrid model. 
 \end{proposition}

\subsection{Wrap, Modulo Reduction, and Share Conversion}\label{subapp:MSWcorsec} 
 First we provide the correctness and security claims for $\mathrm{LiftWrap}.$
 \begin{proposition}\label{prop:LiftWrap}
 The protocol $\mathrm{LiftWrap}$ correctly and securely realizes the functionality $\mathcal{F}_{LiftWrap}$ in the $\mathcal{F}_{GEZ}$-hybrid model.
 \end{proposition}
 \begin{proof}
 The correctness of $\mathrm{LiftWrap}$ follows from the discussion in Section~\ref{MSW}. Under the $\mathcal{F}_{GEZ}$-hybrid model, it is easy to see that $\mathrm{LiftWrap}$ reveals no other information in any steps, proving the security claim of $\mathrm{LiftWrap}.$
 \end{proof}
Next we provide the correctness and security claims for $\mathrm{LiftMod}.$ 
 \begin{proposition}\label{prop:LiftMod}
 The protocol $\mathrm{LiftMod}$ correctly and securely realizes the functionality $\mathcal{F}_{LiftMod}$ in the $\mathcal{F}_{LiftWrap}$-hybrid model.
 \end{proposition}
 \begin{proof}
 The correctness of $\mathrm{LiftMod}$ follows from the discussion in Section~\ref{MSW}. Similarly, as discussed in Section~\ref{MSW}, under the assumption given in $\mathcal{F}_{LiftWrap}$-hybrid model, $\mathrm{LiftMod}$ reveals no other information in any steps, proving the security claim of $\mathrm{LiftMod}.$
 \end{proof} 
 We move on to $\mathrm{DropMod}.$
\begin{proposition}\label{prop:DropMod}
 The protocol $\mathrm{DropMod}$ correctly and securely realizes the functionality $\mathcal{F}_{DropMod}$ in the $(\mathcal{F}_{RndBit},\mathcal{F}_{LiftMod})$-hybrid model with the given security parameter $\kappa.$
 \end{proposition}
 \begin{proof}
First, we prove the correctness of $\mathrm{DropMod}.$ Let $x<N$ and $x_1,\cdots, x_n\in \mathbb{Z}^{N'}$ such that $x_1+\cdots+x_n\equiv x\pmod{N'}.$ It is easy to see that by Propositions~\ref{prop:RndBit} and~\ref{prop:LiftMod}, the first four steps yield a random value $r\in \mathbb{F}_2^{\lceil \log(N)\rceil+\kappa}$ that is secretly shared in two different ways, $[r]^N=(r_1^N,\cdots, r_n^N)$ and $[r]^{N'}=(r_1^{N'},\cdots, r_n^{N'}),$ i.e. $r_1^N+\cdots r_n^N\equiv r \pmod{N}$ and $r_1^{N'}+\cdots+r_n^{N'}\equiv r\pmod{N'}.$ By the choice of value of $N',$ we have $y=x+r.$ Let $y=y_p+\delta_1 N$ where $y_p= y\pmod{N}$ and $r=r_1^N+\cdots+r_n^N+\delta_2 N$ for some integers $\delta_1$ and $\delta_2.$ Then $y_p-(r_1^N+\cdots+r_n^N) = (x+r-\delta_1 N) - (r-\delta_2 N) = x+(\delta_2-\delta_1)N\equiv x \pmod{N}$ as required.

Note that apart from the reveal of the value of $y,$ under the $(\mathcal{F}_{RndBit},\mathcal{F}_{LiftMod})$-hybrid model, $\mathrm{DropMod}$ does not reveal any other values. Hence, by~\cite[Annex A, Theorem $1$]{catrina2010improved}, $\mathrm{DropMod}$ securely realizes $\mathcal{F}_{DropMod}$ with statistical security having security parameter $\kappa.$ 

 \end{proof} 
 
We proceed to the protocol $\mathrm{Wrap}.$

\begin{proposition}\label{prop:Wrap}
 The protocol $\mathrm{Wrap}$ correctly and securely realizes the functionality $\mathcal{F}_{Wrap}$ in the $(\mathcal{F}_{LiftWrap},\mathcal{F}_{DropMod})$-hybrid model with security parameter $\kappa.$
 \end{proposition}
 \begin{proof}
 The correctness of $\mathrm{Wrap}$ directly follows definition and security is obvious since aside from the determination of the values of $\kappa$ and $N'$ which is done independent of the private values, there are no other operations in addition to the calls of the two functionalities.
 \end{proof}
 Lastly, we discuss the protocol $\mathrm{ShaConv}.$
 \begin{proposition}\label{prop:ShaConv}
 The protocol $\mathrm{ShaConv}$ correctly and securely realizes the functionality $\mathcal{F}_{ShaConv}$ in the $(\mathcal{F}_{LiftWrap},\mathcal{F}_{DropMod})$-hybrid model with security parameter $\kappa.$
 \end{proposition}
 \begin{proof}
 The correctness directly follows from the fact that we are calculating Equation~\eqref{deltamod}. Security is guaranteed since no other operations has been done except for the call of the two functionalities and local computations.
 \end{proof}
 \subsection{Beaver Triple Conversion}\label{subapp:BTCcorsec}
 In this section we provide the correctness and security claims for $\mathrm{TripConv}.$
\begin{proposition}\label{prop:TripConv}
 The protocol $\mathrm{TripConv}$ correctly and securely realizes the functionality $\mathcal{F}_{TripConv}$ in the $(\mathcal{F}_{LiftMod},\mathcal{F}_{ShaConv})$-hybrid model with security parameter $\kappa.$
 \end{proposition}
\begin{proof}
The correctness follows from the discussion in Section~\ref{TripleConv} while the security follows from the fact that aside from the determination of the value of $N'$ which is independent of all private values, all calculations are done without revealing any intermediate values and are secure under $(\mathcal{F}_{LiftMod},\mathcal{F}_{ShaConv})$-hybrid assumption.
\end{proof}
\subsection{Probabilistic Bit Generation}\label{subapp:PBGcorsec}
In this section we provide the correctness and security claims for $\mathrm{PrRndBit}.$

\begin{proposition}\label{prop:PBG}
 The protocol $\mathrm{PrRndBit}$ correctly and securely realizes the functionality $\mathcal{F}_{PrRndBit}$ in the $\mathcal{F}_{GEZ}$-hybrid model with security parameter $\kappa.$
 \end{proposition}
\begin{proof}
The correctness follows from the discussion in Section~\ref{TripleConv} while the security follows from the fact that aside from the determination of the value of $N'$ which is independent of all private values, all calculations are done without revealing any intermediate values and are secure under $(\mathcal{F}_{LiftMod},\mathcal{F}_{ShaConv})$-hybrid assumption.
\end{proof}

\end{document}